\documentclass[11pt]{article}

\usepackage[margin=1in]{geometry}
\usepackage{amsthm,amsmath,amsfonts,amssymb,amsrefs}
\usepackage{cancel}
\usepackage[pdftex]{graphicx}
\usepackage{url}
\usepackage[bookmarksnumbered,colorlinks,linkcolor=blue,citecolor=blue,urlcolor=blue]{hyperref}
\usepackage{color} 

\newcommand{\I}{{\mathbb{I}}}
\newcommand{\C}{{\mathbb{C}}}

\newcommand{\R}{{\mathbb{R}}}
\newcommand{\N}{{\mathbb{N}}}
\newcommand{\Z}{{\mathbb{Z}}}

\newcommand{\cP}{{\mathcal{P}}}
\newcommand{\cX}{{\mathcal{X}}}

\newcommand{\erf}{\mathop{\mathrm{erf}}}
\newcommand{\spn}{\mathop{\mathrm{span}}}
\newcommand{\poly}{\mathop{\mathrm{poly}}}

\newcommand{\ii}{{\mathrm{i}}}
\newcommand{\dd}{\mathrm{d}}

\renewcommand{\o}{o}

\newcommand{\be}{\begin{equation}}
\newcommand{\ee}{\end{equation}}
\def\ba#1\ea{\begin{align}#1\end{align}}

\newtheorem{theorem}{Theorem}
\newtheorem{lemma}[theorem]{Lemma}
\newtheorem{corollary}[theorem]{Corollary}
\newtheorem{hypothesis}[theorem]{Hypothesis}

\theoremstyle{definition}
\newtheorem{definition}{Definition}

\begin{document}


\title{Quantum Snake Walk on Graphs}

\author{
\normalsize Ansis Rosmanis\thanks{arosmani@cs.uwaterloo.ca} \\[.5ex]
\small David R.\ Cheriton School of Computer Science \\
\small and Institute for Quantum Computing \\
\small University of Waterloo
}

\date{} 
\maketitle


\begin{abstract}
I introduce a new type of continuous-time quantum walk on graphs called the quantum snake walk, the basis states of which are fixed-length paths (snakes) in the underlying graph. First I analyze the quantum snake walk on the line, and I show that, even though most states stay localized throughout the evolution, there are specific states which most likely move on the line as wave packets with momentum inversely proportional to the length of the snake. Next I discuss how an algorithm based on the quantum snake walk might potentially be able to solve an extended version of the glued trees problem which asks to find a path connecting both roots of the glued trees graph. No efficient quantum algorithm solving this problem is known yet.
\end{abstract}

\section{Introduction}

In 1994, Shor's ground-breaking paper presented an efficient quantum algorithm for integer factoring \cite{Shor1994}. Many other very important results have followed, but the true potential of quantum computers is still unknown. One step towards understanding it is to consider what are the differences between quantum and classical computers in a setting where we are given the input via a black-box oracle. Because we cannot know the inner workings of the oracle, it is easier to prove lower bounds in this model and therefore in many cases it is easier to show a separation between the best quantum and classical algorithms.

Ever since Deutsch presented a simple oracle problem that can be solved on a quantum computer using fewer oracle queries than on any classical computer \cite{Deutsch1}, scientists have tried to come up with more and more problems for which quantum computers outperform their classical counterparts; some very artificial, some quite natural. Bernstein and Vazirani gave the first example of an oracle problem which can be solved in polynomial time on a quantum computer, but requires superpolynomial time on a classical computer \cite{BV93}. Shortly after that Simon gave an example in which this separation is exponential \cite{Simon94}. Yet it is unclear what are the best methods for the construction of efficient quantum algorithms. While quantum Fourier sampling is probably the most popular such method so far, many algorithms are also based on the concept of quantum walk. In particular, continuous-time quantum walks on graphs, introduced by Farhi and Gutmann \cite{Farhi1998}, give rise to fast algorithms for NAND tree evaluation \cite{FarhiGoldstoneGutmann2008} and unstructured search \cite{FG98,ChGst}. Continuous-time quantum walks are also known to be able to perform universal quantum computation \cite{ChildsUniv}, and they can solve some oracle problems exponentially faster than any classical algorithm \cite{CCDFGSart,ChSchVaz}.

In this paper we introduce a new type of continuous-time quantum walk on graphs, the basis states of which are not vertices of the graph, but paths in it of a fixed length $n$. We define this new walk on an unweighted undirected graph $G$ as a regular continuous-time quantum walk (the one introduced by Farhi and Gutmann) on a more complex weighted undirected graph $G_n$ which is constructed from the original graph $G$. Since, to the author of this paper, the discrete-time classical counterpart of this walk in some sense resembles the behavior of a `snake' of length $n$ which is placed on a graph along its edges and which makes random decisions in which direction to go next, we call this new walk the {\em continuous-time quantum snake walk}. 

The main motivation for the continuous-time quantum snake walk is  an extended version of the glued trees problem for which no efficient quantum algorithm yet is known. Childs et al. introduced the glued trees graph, consisting of two complete binary trees of the same height  which are connected by a random cycle that alternates between the leaves of the trees \cite{CCDFGSart}. Given a glued trees graph via a black-box oracle and the label of the root of one tree, the glued trees problem is to determine the label of the root of the other tree. While there is an efficient quantum algorithm based on the continuous-time quantum walk which solves this problem, no efficient quantum algorithm which finds a path connecting these two roots is known.  We hope that the continuous-time quantum snake walk might lead to a quantum algorithm efficiently finding such a path.

In Section \ref{sec:line} we analyze in detail the continuous-time quantum snake walk on the line. Even though most states stay localized throughout the quantum evolution, there are specific states which, under one reasonable assumption, move on the line as wave packets with momentum inversely proportional to the length of the snake. This motion is similar to the regular continuous-time quantum walk on the line. In Section \ref{sec:GTG} we discuss an algorithm based on the continuous-time quantum snake walk that solves the extended version of the glued trees problem mentioned above, that is, finds a path connecting the roots of the glued trees graph. However, it is not clear whether this algorithm runs in polynomial time.

Independently from this work, a recent paper by Mc\,Gettrick defines a discrete-time walk on the line similar to the continuous-time quantum snake walk on the line considered here \cite{McG}.

\section{Definitions and motivation} \label{sec:def}

In this section we define the main concept of this paper, the continuous-time quantum snake walk. The main reason for considering this new type of quantum walk is its potential algorithmic applications. Because of that we also define an extended version of the glued trees problem for which no efficient classical algorithm exists and no efficient quantum algorithm is known, but which might be solved efficiently using an algorithm based on the continuous-time quantum snake walk.

\subsection{Continuous-time quantum snake walk}

We will define the continuous-time quantum snake walk on one graph as the continuous-time quantum walk on another specific graph. Therefore, we define the continuous-time quantum walk first. Let $K=(U,w)$ be a weighted undirected graph, where $U$ is a set of vertices and $w$ is a weight function, which assigns a weight $w(u_1,u_2)\in\R$ to every pair $(u_1,u_2)\in U^2$. Let $\C^U$ be a Hilbert space having an orthonormal basis $\{|u\rangle:u\in U\}$ called the {\em standard basis}, and let $H_K$ be a linear operator acting on $\C^U$ such that $\langle u_1|H_K|u_2\rangle=w(u_1,u_2)$ for all $u_1,u_2\in U$. Because $K$ is undirected with all weights being real, $H_K$ is Hermitian. The {\em continuous-time quantum walk} on $K$ is defined as a quantum evolution in the space $\C^U$ governed by the Hamiltonian $H_K$ according to the Schr\"{o}dinger equation. That is, if we fix an initial state of the walk to be $|\chi(0)\rangle\in\C^U$, then we can look on the walk as a function which maps time $t\geq 0$ to the state $|\chi(t)\rangle=e^{-\ii H_Kt}|\chi(0)\rangle\in\C^U$.

Now suppose $G=(V,E)$ is an undirected unweighted graph, where $V$ is a set of vertices and $E\subset V^2$ is a set of edges. Let $S_n(G)$ be the set of all paths in $G$ which have length $n$. Here we assume that a path can visit a vertex multiple times. Let us call an element of $S_n(G)$ a {\em snake} of length $n$. That is, a snake of length $n$ is a vector $s=(v_0,\ldots,v_n)\in V^{n+1}$ such that $(v_{l-1},v_l)\in E$ for all $l\in[1\,..\,n]$. Note that we assume all paths are directed ($(v_0,v_1,\ldots,v_n)$ and $(v_n,v_{n-1},\ldots,v_0)$ are not the same snake), even though the underlying graph $G$ is undirected.

Let $s=(v_0,\ldots,v_n)\in S_n(G)$ be a snake of length $n$. We say that $s$ can move forward to a snake $t\in S_n(G)$ if there exists a vertex $v_{n+1}\in V$ such that $t=(v_1,\ldots,v_{n+1})$, and we write $m_f(s,t)$ (we think of $m_f$ as a predicate). Similarly, we say that $s$ can move backward to $t\in S_n(G)$ if there exists $v_{-1}\in V$ such that $t=(v_{-1},\ldots,v_{n-1})$, and we write $m_b(s,t)$. We also consider $m_f(s,t)$ as a binary function on $S_n(G)^2$ taking the value $1$ if and only if the predicate $m_f$ is true for a pair $(s,t)$, and similarly for $m_b$. Let $a_{s,t}=m_f(s,t)+m_b(s,t)\in\{0,1,2\}$, and let $A_n(G)$ be a matrix whose rows and columns are labeled by the elements of $S_n(G)$ and $A_n(G)_{s,t}=a_{s,t}$ for all $s,t\in S_n(G)$. Clearly $m_f(s,t)$ if and only if $m_b(t,s)$, and therefore $A_n(G)$ is a symmetric matrix. We can look on $A_n(G)$ as the adjacency matrix of a graph having $S_n(G)$ as the set of vertices, and possibly some edges having weight $2$ instead of $1$ (as a matter of fact, $A_n(G)_{s,t}=2$ if and only if there are two adjacent vertices $u$ and $v$ such that $s=(u,v,u,v,\ldots)$ and $t=(v,u,v,u,\ldots)$).

\begin{definition}
Let $G_n$ be a weighted graph with the set of vertices $S_n(G)$ and the matrix of weights $A_n(G)$. The {\em continuous-time quantum snake walk} on the graph $G$ is defined as a continuous-time quantum walk on the weighted graph $G_n$.
\end{definition}

Because we consider only continuous-time walks, we often refer to continuous-time quantum walks and continuous-time quantum snake walks, respectively, as regular quantum walks and quantum snake walks. Sometimes we use the word snake to refer, instead of an element of $S_n(G)$, to the content of a quantum register corresponding to the Hilbert space $\C^{S_n(G)}$. In some sense, we can think of the snake as being a path which can move form one position in the graph to another. In what sense the word snake is used should be clear from the context.

\subsection{Extended glued trees problem} \label{ssec:extdef}

In order to provide an example of a quantum algorithm based on a quantum walk that can solve a certain black-box problem exponentially faster than any classical algorithm, Childs et al.\ introduced the glued trees graph \cite{CCDFGSart}. The {\em glued trees graph} consists of two complete binary trees of height $N$, which are connected by a random cycle that alternates between the leaves of the two trees. An example of the graph is shown in Figure \ref{fig:gluedtrees} (for $N=3$).
Let us also call $N$ and the roots of both trees, respectively, the height and the roots of the glued trees graph.

\begin{figure}[htp]
\centering
\includegraphics[scale=0.75]{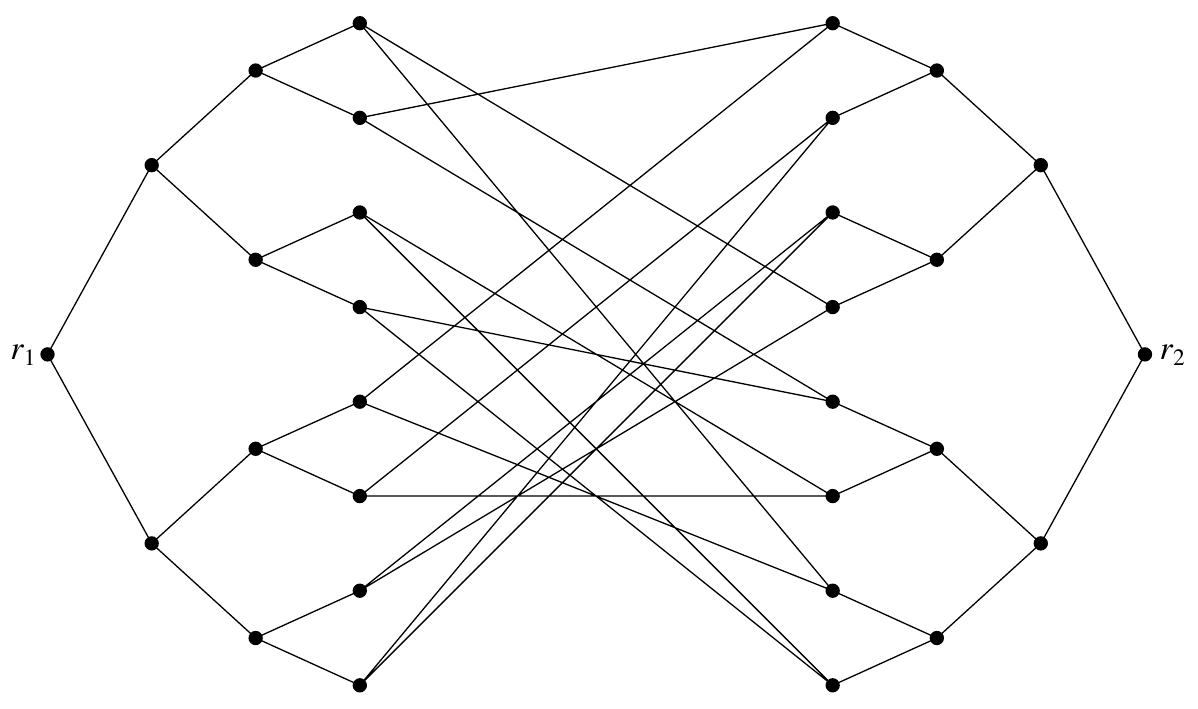}
\caption{A glued trees graph of height 3.}
\label{fig:gluedtrees}
\end{figure}

Suppose we are given a glued trees graph via a black-box oracle (given the label of a vertex as an input, the oracle outputs labels of all its neighbors) and the label of one root $r_1$. The {\em glued trees problem} is to determine the label of the other root $r_2$ (see  Figure \ref{fig:gluedtrees}). Childs et al.\ show that, if we start a continuous-time quantum walk at the root $r_1$, the walk quickly (in $\poly N$ time) traverses the graph and ends up in a superposition state which has a large overlap on the root $r_2$. We can implement continuous-time quantum walks on sparse graphs in the quantum circuit model efficiently (see \cite{ATart,ChildsThesis,BeAhClSa}), therefore there is an efficient quantum query algorithm solving the glued trees problem. It is also shown in \cite{CCDFGSart} that no classical query algorithm can solve this problem efficiently.

\begin{definition}
Given a glued trees graph via a black-box oracle and the label of one of its two roots $r_1$, the {\em extended glued trees problem} is to find a path in the graph connecting $r_1$ to the opposite root $r_2$.
\end{definition}

As far as I know, there is still no quantum algorithm known which solves the extended glued trees problem efficiently. It is not even known if such an efficient algorithm exists \cite[Section 5.4]{ChildsThesis}. However, the hope is that the quantum snake walk might lead to an efficient algorithm.

\section{Quantum snake walk on the line}   \label{sec:line}

As an example of a continuous-time quantum snake walk, let us consider the quantum snake walk on the line, that is, the graph $G=(V,E)$, where $V=\Z$ is the set of vertices and $E=\{(x,x\pm1)\,:\,x\in\Z\}$ is the set of edges. This example is relatively simple compared to quantum snake walks on other graphs. Nonetheless, understanding this walk later helps us to analyze snake walks on more complex graphs.

\subsection{The Hamiltonian}

It is convenient to think of the line also as the $X$-axis. For every snake $(v_0,\ldots,v_n)\in V^{n+1}$ on $G$ of length $n$ let $x=v_0$ be the start vertex of the snake and for $l\in[1\,..\,n]$ let $j_l\in\{0,1\}$ be such that $v_l=v_{l-1}-(-1)^{j_l}$. This gives a one-to-one relation between the set of snakes $S_n(G)$ and the set $\Z\times\{0,1\}^n$. From now on let us consider any snake $s$ to be given as a pair $(x,j_1\ldots j_n)\in\Z\times\{0,1\}^n$, and let $|x\rangle|j_1\ldots j_n\rangle$ denote $|s\rangle$.

Every snake on the line can move forward to two other snakes and move backward to two other snakes. If a snake $(x,j_1\ldots j_n)$ moves forward, then the start vertex of the new snake is determined by $j_1$. Its end vertex can either move in the positive or the negative direction of the $X$-axis. If a snake moves backward, the direction in which its start vertex moves is opposite to the direction of the first edge of the new snake. In other words, for every $j\in\{0,1\}^n$ we have
\begin{equation*}
\begin{split}
   m_f((x,j_1\ldots j_n),(x-(-1)^{j_1},j_2\ldots j_n0))\quad &\text{and}\quad
   m_f((x,j_1\ldots j_n),(x-(-1)^{j_1},j_2\ldots j_n1)), \\
   m_b((x,j_1\ldots j_n),(x-1,1j_1\ldots j_{n-1}))\quad &\text{and}\quad
   m_b((x,j_1\ldots j_n),(x+1,0j_1\ldots j_{n-1})).
\end{split}
\end{equation*}
This allows us to obtain the Hamiltonian $H_n=A_n(G)$ governing the quantum snake walk on the line. After some derivations we get
\begin{equation*}
H_n = \int_0^{2\pi}{|\tilde{k}\rangle\langle\tilde{k}|\otimes H_{n,k}\,\dd k},
\end{equation*}
where
\begin{equation*}
|\tilde{k}\rangle=\frac{1}{\sqrt{2\pi}}\sum_{x\in\Z}{e^{\ii kx}|x\rangle}
\end{equation*}
 and 
\begin{equation}
\label{eq:lineHnk}
\begin{split} 
  H_{n,k} \; = \; & \sum_{j\in\{0,1\}^{n-1}} {e^{\ii k}\;(\,|j0\rangle\langle 0j|+|j1\rangle\langle 0j|+|1j\rangle\langle j0|+|1j\rangle\langle j1| \,)} \\
             +    & \sum_{j\in\{0,1\}^{n-1}} {e^{-\ii k} (\,|0j\rangle\langle j0|+|0j\rangle\langle j1|+|j0\rangle\langle 1j|+|j1\rangle\langle 1j| \,)}
\end{split}
\end{equation}
(see \cite[Section 3.1]{RosmanisThesis} for details). For every $k_1,k_2\in\R$ we have $\langle\tilde{k_2}|\tilde{k_1}\rangle=\Delta(k_1-k_2\!\!\mod 2\pi)$,  where $\Delta$ is the Dirac delta function (see \cite{ChildsUniv}). Hence, for any value $k$, if $|\psi\rangle$ is an eigenvector of $H_{n,k}$ with the corresponding eigenvalue $\lambda$, then $|\tilde{k}\rangle\otimes|\psi\rangle$ is an eigenvector of $H_n$ corresponding to the same eigenvalue $\lambda$. 

Let us focus on finding the eigenvalues and eigenvectors of $H_{n,k}$ for an arbitrary $k$. This task becomes much easier if we express $H_{n,k}$ in a different basis. Consider two pairs of orthonormal vectors:
\begin{equation*}
\begin{split} 
& |u_{0,k}\rangle=\frac{1}{\sqrt{2}}\left(e^{-\ii k}|0\rangle+e^{\ii k}|1\rangle\right) \quad\text{and}\quad
|u_{1,k}\rangle=\frac{1}{\sqrt{2}}\left(e^{-\ii k}|0\rangle-e^{\ii k}|1\rangle\right), \\
& |v_0\rangle=\frac{1}{\sqrt{2}}\left(|0\rangle+|1\rangle\right) \quad\text{and}\quad
|v_1\rangle=\frac{1}{\sqrt{2}}\left(|0\rangle-|1\rangle\right). 
\end{split} 
\end{equation*}
 For $m\in[0\,..\,2^n-1]$ let us define a unit vector $|\widehat{m}_k\rangle$ as follows. First, let $|\widehat{0}_k\rangle=-\ii |u_{0,k}\rangle^{\otimes n}$. For $m\in[1\,..\,2^n-1]$, let $m$ written in binary using $\lfloor\log_2(m)\rfloor+1$ bits be $1m_{\lfloor \log_2(m) \rfloor}\ldots m_1$. Then for $m\in[1\,..\,2^n-1]$ we define
\begin{equation*}
  |\widehat{m}_k\rangle  =  |u_{0,k}\rangle^{\otimes n-\lfloor\log_2(m)\rfloor-1}  |u_{1,k}\rangle  |v_{m_{\lfloor \log_2(m)\rfloor}}\rangle  \ldots  |v_{m_1}\rangle.
\end{equation*}
As an example, for $n=3$ we have
  $|\widehat{0}_k\rangle=-\ii |u_{0,k}\rangle|u_{0,k}\rangle|u_{0,k}\rangle$,
  $|\widehat{1}_k\rangle=|u_{0,k}\rangle|u_{0,k}\rangle|u_{1,k}\rangle$,
  $|\widehat{2}_k\rangle=|u_{0,k}\rangle|u_{1,k}\rangle|v_{0}\rangle$,
  $|\widehat{3}_k\rangle=|u_{0,k}\rangle|u_{1,k}\rangle|v_{1}\rangle$,
  $|\widehat{4}_k\rangle=|u_{1,k}\rangle|v_{0}\rangle|v_{0}\rangle$,
  $|\widehat{5}_k\rangle=|u_{1,k}\rangle|v_{0}\rangle|v_{1}\rangle$,
  $|\widehat{6}_k\rangle=|u_{1,k}\rangle|v_{1}\rangle|v_{0}\rangle$ and
  $|\widehat{7}_k\rangle=|u_{1,k}\rangle|v_{1}\rangle|v_{1}\rangle$.
For any $n\in\N$ and $k\in\R$ the set $B_{n,k}=\{ |\widehat{0}_k\rangle, \ldots, |\widehat{2^n-1}_k\rangle \}$ is an orthonormal basis of $\C^{\{0,1\}^n}$, and \cite[Section 3.1.1]{RosmanisThesis} shows that $H_{n,k}$ expressed in this basis is
\begin{equation}
\label{eq:HnkInHatBasis}
  H_{n,k}=
  4 \cos k \, |\widehat{0}_k\rangle\langle\widehat{0}_k|+
   2 \sin k \, (|\widehat{1}_k\rangle\langle\widehat{0}_k| + |\widehat{0}_k\rangle\langle\widehat{1}_k|)+
   2 \sum_{m=1}^{2^{n-1}-1}{(|\widehat{2m}_k\rangle\langle\widehat{m}_k|+|\widehat{m}_k\rangle\langle\widehat{2m}_k|)}.
\end{equation}

From (\ref{eq:HnkInHatBasis}) we see that $\langle\widehat{m}_k|H_{n,k}|\widehat{l}_k\rangle=0$ whenever $m$ and $l$ have different greatest odd divisors, where we assume that the greatest odd divisor of $0$ is $1$. Therefore we can block diagonalize $H_{n,k}$ with respect to $B_{n,k}$ into $2^{n-1}$ blocks.
To be more precise, let us define $2^{n-1}$ orthogonal projectors 
$\Pi_{1,k}=|\widehat{0}_k\rangle\langle\widehat{0}_k|+\sum_{j=0}^{n-1}{|\widehat{2^j}_k\rangle\langle\widehat{2^j}_k|}$ and {\label{PiProjector} $\Pi_{l,k}=\sum_{j=0}^{n-\lceil\log_2 l\rceil}{|\widehat{l\cdot2^j}_k\rangle\langle\widehat{l\cdot2^j}_k|}$} for odd $l\in[3\,..\,2^n-1]$. We have $\Pi_{l_1,k}H_{n,k}\Pi_{l_2,k}=0$ whenever $l_1\neq l_2$. Therefore, in order to get the full eigenspectrum of $H_{n,k}$, we consider the eigenvalues of each block separately.

For odd $l\in[3\,..\,2^n-1]$ the block
\begin{equation*}
\Pi_{l,k}H_{n,k}\Pi_{l,k}
=2\sum_{j=0}^{n-\lceil\log_2 l\rceil-1}
{\left(|\widehat{l\cdot2^{j+1}}_k\rangle\langle\widehat{l\cdot2^j}_k|+|\widehat{l\cdot2^j}_k\rangle\langle\widehat{l\cdot2^{j+1}}_k|\right)},
\end{equation*}
is basically twice the adjacency matrix of the line segment of length $n-\lceil\log_2 l\rceil$, whose eigenvalues and eigenvectors are well known (for example, see \cite{AmbainisLNotes}). Therefore the eigenvalues of this block are $k$-independent and the eigenvectors depend only on $k$ due to the fact that the basis $B_{n,k}$ is $k$-dependent. However, because this dependence is relatively simple, as shown in \cite[Section 3.1.3]{RosmanisThesis}, the following lemma holds.

\begin{lemma} \label{the:kindependent}
Let $x_1,x_2\in\Z$ and $j_1,j_2\in\{0,1\}^n$. If $|x_1-x_2|>2n$, then
$\langle x_2,j_2|e^{-\ii H_nt}|x_1,j_1\rangle = \langle x_2,j_2|e^{-\ii K_nt}|x_1,j_1\rangle$,
where
$K_n=\int_{0}^{2\pi}|\tilde{k}\rangle\langle\tilde{k}|\otimes \Pi_{1,k}H_{n,k}\Pi_{1,k}\, \dd k$.
\end{lemma}

The essence of Lemma \ref{the:kindependent} is that we can restrict our attention to the states in the space corresponding to the projector $\int_{0}^{2\pi}|\tilde{k}\rangle\langle\tilde{k}|\otimes \Pi_{1,k}\, \dd k$ if we want to have an initial state which moves on the line further than distance $2n$, which  for algorithmic interests seems to be a reasonable request. The eigenvalues and eigenvectos of $H_n$ restricted to this subspace, namely $K_n$, is determined by the eigenvalues and eigenvectors of the block $\Pi_{1,k}H_{n,k}\Pi_{1,k}$.

Let $\{|\overline{1}\rangle,\ldots,|\overline{n+1}\rangle\}$ be a fixed set of $n+1$ orthonormal, $k$-independent vectors. Consider a linear $k$-dependent isometry
\begin{equation*}
U_{n,k}=\sum_{y=1}^{n}{|\widehat{2^{n-y}}_k\rangle\langle\overline{y}|+|\widehat{0}_k\rangle\langle\overline{n+1}|}.
\end{equation*}
We have $\Pi_{1,k}H_{n,k}\Pi_{1,k}=U_{n,k}\Phi_{n,k}U_{n,k}^*$, where
\begin{equation}
\label{eq:Phik}
\Phi_{n,k}
   = 2 \sum_{y=1}^{n-1}
              {(|\overline{y+1}\rangle\langle\overline{y}|+|\overline{y}\rangle\langle\overline{y+1}|)}
    + 2\sin k \,(|\overline{n+1}\rangle\langle\overline{n}|+|\overline{n}\rangle\langle\overline{n+1}|)
    + 4\cos k \, |\overline{n+1}\rangle\langle\overline{n+1}|,
\end{equation}
which is the adjacency matrix of the graph given in Figure \ref{fig:PhiGraph}.

\begin{figure}[htp]
\centering
\includegraphics[scale=1]{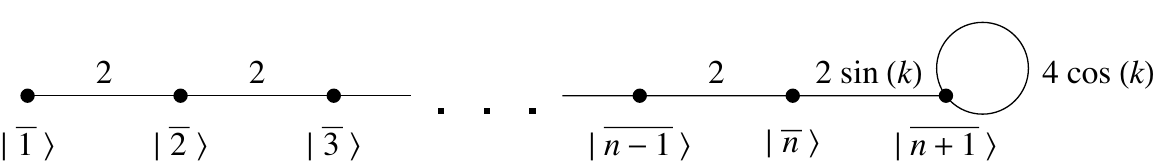}
\caption{The graph corresponding to the adjacency matrix $\Phi_{n,k}$.}
\label{fig:PhiGraph}
\end{figure}

If $k\equiv 0\mod\pi$, then $\sin k=0$ and thus $\Phi_{n,k}$ is the sum of two orthogonal operators $2\sum_{y=1}^{n-1}(|\overline{y+1}\rangle\langle \overline{y}|+|\overline{y}\rangle\langle \overline{y+1}|)$ and $4\cos k|\overline{n+1}\rangle\langle\overline{n+1}|$. The eigenspecturm of the former operator is $2$ times the eigenspectrum of the adjacency matrix of the line segment of length $n-1$, and the eigenspecturm of the latter is trivial. Therefore, in the case when $k\equiv 0\mod\pi$ we can get all the eigenvalues and eigenvectors of $\Phi_{n,k}$: for every $p\in\{\frac{\pi}{n+1},\frac{2\pi}{n+1},\ldots,\frac{n\pi}{n+1}\}$ the unit vector $\sqrt{\frac{2}{n+1}}\sum_{y=1}^{n}{\sin yp\,|\overline{y}\rangle}$ is an eigenvector of $\Phi_{n,k}$ with the corresponding eigenvalue $4\cos p$, and $|\overline{n+1}\rangle$ is an eigenvector corresponding to the eigenvalue $4\cos k$ (i.e., $\pm4$).

Unfortunately, in the general case we do not even know if there are closed form expressions for the eigenvalues of $\Phi_{n,k}$. However, the following lemma, which is proved in Appendix \ref{app:pcondition}, gives us some useful information about the eigenvalues and eigenvectors of $\Phi_{n,k}$ for an arbitrary $k$.

\begin{lemma} \label{the:pcondition}
Let us fix $n\in\N$ and $k\:\cancel{\equiv}\:0\!\!\mod\pi$. The equation
\begin{equation} \label{eq:pcondition}
  2(\cos p -\cos k)\sin ((n+1)p) = \sin^2 k \sin np
\end{equation}
has $n+1$ distinct solutions in the interval $(0,\pi)$, and, if $p$ is  a solution of this equation, then $\sum_{y=1}^n{\sin yp|\overline{y}\rangle}+\frac{\sin (n+1)p}{\sin k}|\overline{n+1}\rangle$ is an (unnormalized) eigenvector of $\Phi_{n,k}$ with the corresponding eigenvalue $4\cos p$.
\end{lemma}

\begin{corollary}
The eigenvalues of $\Phi_{n,k}$ are all distinct for any value of $k$, and they all depend on the value of $k$.
\end{corollary}

We refer to (\ref{eq:pcondition}) as the {\em $p$-equation}.

\subsection{Even {\em n} and the median eigenvalue}

We do not know how to solve the $p$-equation, the solution of which would give us the full eigenspectrum of $H_{n,k}$. Despite that, the $p$-equation allows us to obtain good approximations of eigenvalues and their derivatives. But let us first show that we can treat $k$-dependent eigenvalues of $H_{n,k}$ as infinitely differentiable functions of $k$.

A complex-valued function is {\em holomorphic} if it is complex-differentiable in a neighborhood of every point in its domain. Holomorphic functions are known to be infinitely differentiable, and the class of holomorphic functions include polynomials, the exponential function, sine and cosine \cite{Markushevich}. The notion of holomorphic functions can be generalized to vector-valued and operator-valued functions in an obvious way. From \cite[Chapter  II, \S6.2]{Kato} we have the following lemma

\begin{lemma}
\label{lem:holo}
Let $\cX$ be a finite complex Euclidean space of dimension $m$. Consider a holomorphic operator-valued function $T$ which maps complex numbers to linear operators over $\cX$ such that $T(k)$ is Hermitian for all $k\in\R$. Then for $k\in\R$ there exist a family of orthonormal basis $\{|\phi_l(k)\rangle\,:\,l\in[1\,..\,m]\}$ of $\cX$ consisting of eigenvectors of $T(k)$ and a family $\{\lambda_l(k)\,:\,l\in[1\,..\,m]\}$ consisting of eigenvalues of $T(k)$ such that $\lambda_l(k)$ and $|\phi_l(k)\rangle$ are holomorphic
functions of $k$ for all $l\in[1\,..\,m]$.
\end{lemma}

Let $\lambda_l(k)$ be the $l$-th largest eigenvalue of $\Phi_{n,k}$. Because all the eigenvalues of $\Phi_{n,k}$ are distinct and $\Phi_{n,k}$ is holomorphic in $k$ (which can be easily seen from (\ref{eq:Phik})), according to Lemma \ref{lem:holo}, $\lambda_l(k)$ is a holomorphic function in $k$ for all $l\in[1\,..\,n+1]$. In order to give some intuition about the eigenspectrum of $\Phi_{n,k}$, in Figure \ref{fig:eigenvalues} we show how the eigenvalues of $\Phi_{n,k}$ and their derivatives depend on $k\in[0,2\pi]$ in the case when $n=8$. These plots are obtained via numerical computation.

\begin{figure}[htp]
\centering
\includegraphics[scale=1]{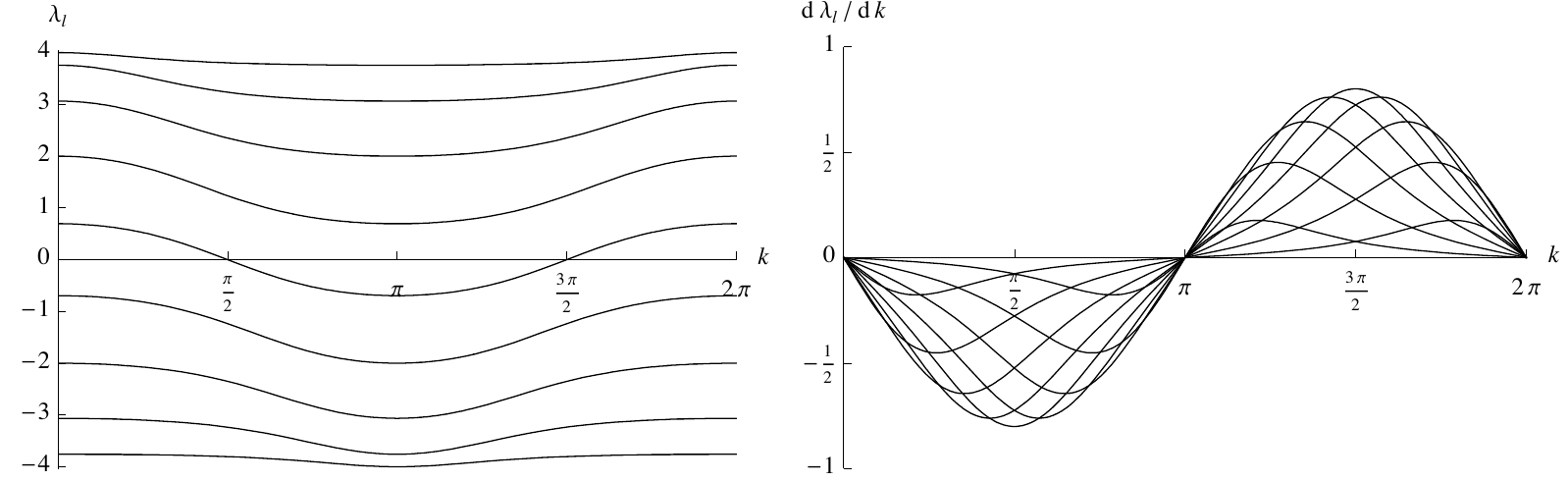}
\caption{The eigenvalues of $\Phi_{n,k}$ (on the left) and their derivatives (on the right) for $n=8$.}
\label{fig:eigenvalues}
\end{figure}

Figure \ref{fig:eigenvalues} suggests that the derivative of any eigenvalue is $0$ if and only if $k\equiv0\mod\pi$. The following theorem, which is proved in \cite[Section 3.2.1]{RosmanisThesis} using the $p$-equation, confirms this observation.

\begin{theorem}
\label{thm:lambdaD}
For every $l\in[1\,..\,n+1]$ the eigenvalue function $\lambda_l(k)$ is strictly decreasing in the interval $(0,\pi)$ and strictly increasing in the interval $(\pi,2\pi)$.
\end{theorem}

For the rest of the section we consider only the case when $n$ is even because it makes some calculations easier. Also, out of all $n+1$ $k$-dependent eigenvalues of $H_{n,k}$, let $\lambda(k)$ be the median one, that is, the $\frac{n+2}{2}$-th largest; we consider only this particular eigenvalue. The reasons for considering this particular eigenvalue are the fact that it can be well approximated and it seems to have the largest derivative, which most likely is good for a construction of states moving fast on the line.

Note that the value of $\lambda(k)$ depends on $n$. Let $\Lambda(k)=4\arctan\left(\frac{2\cos k}{\sin^2 k}\right)$ and let us use a prime to denote derivatives with respect to $k$. In Appendix \ref{app:hyplem} we prove the following lemma.
\begin{lemma}
\label{lem:lambdaprime}
For every $n$, $\lambda'(k)$ is bounded between $\frac{\Lambda'(k)}{n}\left(1-\frac{2}{n}\right)$ and $\frac{\Lambda'(k)}{n}\left(1+\frac{2}{n}\right)$ and we have  $\left|\lambda''(k)-\frac{\Lambda''(k)}{n}\right|\in O\left(\frac{1}{n^2}\right)$. Also, there exists $n_0\in\N$ such that for all $n\geq n_0$ we have $\lambda''(k)=0$ if and only if $k\equiv\frac{\pi}{2}\mod\pi$.
\end{lemma}
\begin{corollary}
The value of $\lambda(k)$ is bounded between $\frac{\Lambda(k)}{n}\left(1-\frac{2}{n}\right)$ and $\frac{\Lambda(k)}{n}\left(1+\frac{2}{n}\right)$.
\end{corollary}
As $n$ increases, the function $n\lambda(k)$ and its derivatives $n\lambda'(k)$ and $n\lambda''(k)$ converge to the functions  $\Lambda(k)$, $\Lambda'(k)$ and $\Lambda''(k)$, respectively. In Figure \ref{fig:arctan} we plot $\Lambda(k)$ and its first two derivatives to get some intuition about the behaviour of this `eigenvalue' function.

\begin{figure}[htp]
\centering
\includegraphics[scale=0.7]{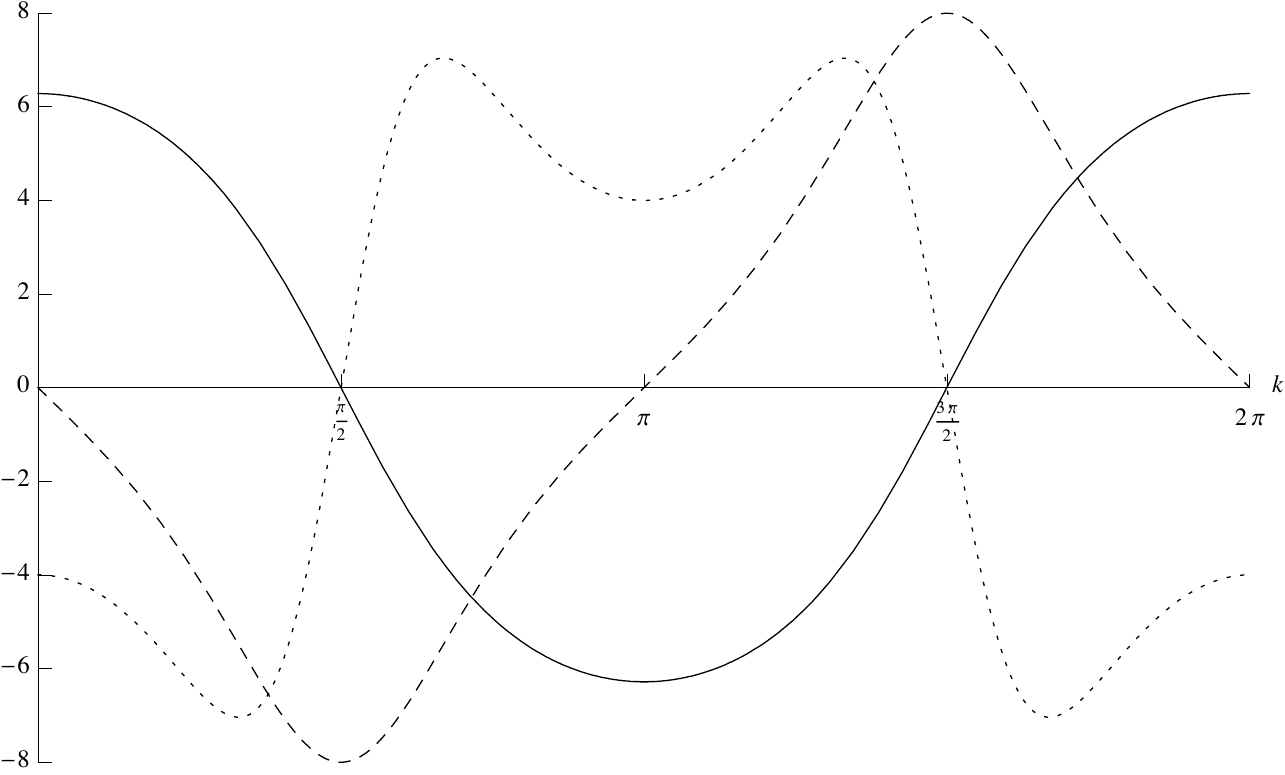}
\caption{$\Lambda(k)=4\arctan\left(\frac{2\cos k}{\sin^2 k}\right)$ (solid line) and its first and second derivatives (dashed and dotted lines, respectively).}
\label{fig:arctan}
\end{figure}

\subsection{A localized superposition of `median' eigenvectors evolving as wave packets} \label{ssec:linewave}

Now that we have established some properties of the median eigenvalue $\lambda(k)$, let us consider a quantum snake walk on the line having an initial state $|\eta_0\rangle$ which is a uniform superposition of eigenvectors of $H_n$ corresponding to $\lambda(k)$. We will choose the phases of these eigenvectors in a particular manner so that $|\eta_0\rangle$ seems to be highly localized. Then we will show that, under one reasonable assumption, for $n$ large enough the state $|\eta_0\rangle$ evolves as a left-moving and a right-moving wave packet, each propagating with the same constant momentum. But let us first give a precise definition of $|\eta_0\rangle$.

Lemma \ref{the:pcondition} implies that there exists a unique vector-valued function $|\phi(\cdot)\rangle:\R\rightarrow\R^{n+1}$ such that  $\Phi_{n,k}|\phi(k)\rangle=\lambda(k)|\phi(k)\rangle$, $\langle\phi(k)|\phi(k)\rangle=1$ and  $\langle\overline{1}|\phi(k)\rangle>0$ for all $k\in\R$. And Lemma $\ref{lem:holo}$ implies that $|\phi(k)\rangle$ is continuous in $k$, and therefore so is $|\psi(k)\rangle=U_{n,k}|\phi(k)\rangle$. Hence, $|\tilde{k}\rangle\otimes|\psi(k)\rangle$ is an eigenvector of $H_n$ corresponding to the eigenvalue $\lambda(k)$. For every $x\in\Z$ define
\begin{equation}
\label{eq:defeta}
|\eta_x\rangle = \frac{1}{\sqrt{2\pi}}
\int_{0}^{2\pi}
e^{-\ii k x}|\tilde{k}\rangle\otimes|\psi(k)\rangle\,\dd k.
\end{equation}
Notice that $|\eta_x\rangle=\sum_{y\in\Z}|y+x\rangle\langle y|\cdot|\eta_0\rangle$.
The state $|\eta_0\rangle$ obviously belongs to the subspace $\cP=\spn(\,\{\,|\tilde{k}\rangle\otimes|\psi(k)\rangle\,:\,k\in[0,2\pi]\,\}\,)$ and remains in $\cP$ under the evolution of $H_n$. Because $\langle\eta_x|\eta_y\rangle=\delta_{x,y}$, where $\delta$ is the Kronecker delta function, and
\begin{equation*}
\frac{1}{\sqrt{2\pi}}\sum_{x\in\Z}e^{\ii lx}|\eta_x\rangle
  =
\int_{0}^{2\pi}
\left( \frac{1}{2\pi} \sum_{x\in\Z} e^{\ii(l-k)x} \right)
|\tilde{k}\rangle\otimes|\psi(k)\rangle\,\dd k 
   = 
   |\tilde{l}\rangle\otimes|\psi(l)\rangle
\end{equation*}
for any $l\in\R$, $\{|\eta_x\rangle:x\in\Z\}$ is an orthonormal basis of $\cP$.

Numerical results suggest that $|\eta_x\rangle$ is highly localized `around the value $x$ on the $X$-axis' (see Appendix \ref{app:linenum}). For example, for $n=14$, if we measure $|\eta_x\rangle$ in the standard basis, with probability $0.62$ we obtain a snake starting at the position $x\pm1$, with probability $0.26$ a snake starting at $x\pm3$, with $0.09$ a snake starting at $x\pm5$, and so on (for even $y$ we never obtain a snake starting at position $x+y$). Also, the probability of obtaining a snake starting at the position at least $n$ units away from $x$ seems to decrease exponentially with $n$, and the following hypothesis seems to hold.

\begin{hypothesis}
\label{hyp:local}
For every $y\in\Z$ and $j\in\{0,1\}^n$ let $a_{y,j}=\langle y,j|\eta_x\rangle$. Then for every polynomial $r(n)$ we have
\begin{equation*}
\sum_{y=-\infty}^{x-n}\sum_{j\in\{0,1\}^n}|\alpha_{y,j}|+\sum_{y=x+n}^{+\infty}\sum_{j\in\{0,1\}^n}|\alpha_{y,j}|\in\o(1/r(n)).
\end{equation*}
\end{hypothesis}

Basically, Hypothesis \ref{hyp:local} implies that for any $|\chi\rangle\in\cP$, if $|\langle\eta_x|\chi\rangle|$ is large, then for some $y\in[x-n+1\,..\,x+n-1]$ the state $|\chi\rangle$ has large overlap with the space of snakes starting at the position $y$; meanwhile, if $|\langle\eta_y|\chi\rangle|$ is small for all $y\in[x-n+1\,..\,x+n-1]$, then the state $|\chi\rangle$ has small overlap with the space of snakes starting at the position $x$. Therefore, the inner products between $|\chi\rangle$ and vectors of the basis $\{|\eta_x\rangle:x\in\Z\}$ allow us to specify where on the line the state $|\chi\rangle$ is located.

Let us analyze how the state $|\eta_0\rangle$ evolves under the Hamiltonian $H_n$ in the limit of an asymptotically large time of the evolution $t$. That is, for every $\omega\in\R$ let us consider $\langle \eta_{\omega t}|e^{-\ii H_n t}|\eta_0\rangle$ for all $t\geq0$ such that $\omega t\in\Z$, and let us think of $\omega$ as the momentum. We have the following:
\begin{equation}
\label{eq:xiz1xiz2Int}
\langle\eta_{\omega t}|e^{-\ii H_n t}|\eta_0\rangle
=  \frac{1}{2\pi}
  \int_{0}^{2\pi}
  e^{\ii t (\omega k - \lambda(k))}\,\dd k.
\end{equation}
Let $f_\omega(k)=\omega k - \lambda(k)$. When $t$ is large, $e^{\ii t(\omega k - \lambda(k))}=e^{\ii t f_\omega(k)}$ rapidly changes as $f_\omega(k)$ changes. This means that in an interval where $f_\omega(k)$ changes the contribution from adjacent subintervals to the integral (\ref{eq:xiz1xiz2Int}) nearly cancels out \cite{NVv1}. Therefore the most contribution to the integral comes from values of $k$ where the function $f_\omega(k)$ is stationary, that is, where $f_\omega'(k)=0$. The method of stationary phase approximation makes this statement more rigorous, and we will use it to estimate the integral  (\ref{eq:xiz1xiz2Int}) in the limit of asymptotically large $t$.

From the $p$-equation one can get that $\lambda'(\frac{\pi}{2})=-\frac{8}{n+2}$ and $\lambda'(\frac{3\pi}{2})=\frac{8}{n+2}$, $\lambda^{(3)}(\frac{\pi}{2})=\frac{8(3n^2+4)}{(n+2)^3}$ and $\lambda^{(3)}(\frac{3\pi}{2})=-\frac{8(3n^2+4)}{(n+2)^3}$. Let us consider $n\geq n_0$, where $n_0$ is the same as in Lemma \ref{lem:lambdaprime}. Lemma \ref{lem:lambdaprime} implies that the derivative $\lambda'(k)$ is strictly increasing in the interval $(\frac{\pi}{2},\frac{3\pi}{2})$. For each $\omega\in(-\frac{8}{n+2},\frac{8}{n+2})$ let $k_\omega$ be the unique point in the interval $(\frac{\pi}{2},\frac{3\pi}{2})$ such that $f_\omega(k_\omega)=0$. Notice that for $\omega$ such that $|\omega|>\frac{8}{n+2}$ the function $f_\omega(k)$ has no stationary points. The method of stationary phase approximation (see \cite{AmbainisLNotes,NVv1}) applied to the integral (\ref{eq:xiz1xiz2Int}) gives us that, in the limit of asymptotically large $t$,
\begin{itemize}
  \item if $|\omega|>\frac{8}{n+2}$, then
    \begin{equation*}
      \langle\eta_{\omega t}|e^{-\ii H_n t}|\eta_0\rangle
   \in o\left(\frac{1}{r(t)}\right)
    \end{equation*}
    for every polynomial $r(t)$;
    \item if $\omega=\pm\frac{8}{n+2}$, then
    \begin{equation*}
      \langle\eta_{\pm\frac{8}{n+2} t}|e^{-\ii H_n t}|\eta_0\rangle
   \approx
    \ii^{-\frac{8}{n+2}t}\frac{(n+2)\Gamma(1/3)}{2\pi\sqrt[3]{4\sqrt{3}(3n^2+4)}} \cdot \frac{1}{\sqrt[3]{t}};
    \end{equation*}
    \item if $\omega\in(-\frac{8}{n+2},\frac{8}{n+2})$, then 
     \begin{equation*}
      \langle\eta_{\omega t}|e^{-\ii H_n t}|\eta_0\rangle
   \approx
   \ii^{\omega t}\sqrt{\frac{2}{\pi}}\cos\left(
   t\left(\omega\left(k_\omega-\frac{\pi}{2}\right)-\lambda(k_\omega)\right)-\frac{\pi}{4}
   \right) \cdot \frac{1}{\sqrt{t\;|\lambda''(k_\omega)|}}.
    \end{equation*}
\end{itemize}
As we can see, if we measure $e^{-\ii H_n t}|\eta_0\rangle$ in the basis $\{|\eta_{\omega t}\rangle:\omega t\in\Z\}$, most likely we will obtain a state $|\eta_{\omega t}\rangle$ such that $\lambda''(k_\omega)$ is close to $0$. According to Lemma \ref{lem:lambdaprime}, $\lambda''(k_\omega)$ is close to $0$ when $k_\omega$ is close to $\frac{\pi}{2}$ and $\frac{3\pi}{2}$ (see Figure \ref{fig:arctan}) and therefore $\omega$, the momentum, is close to $\pm\frac{8}{n+2}$. Hence, if indeed the state $|\eta_{\omega t}\rangle$ is localized around the position $\omega t$ on the $X$-axis, as Hypothesis \ref{hyp:local} asserts, then we can see that for $n\geq n_0$ the state $|\eta_0\rangle$ evolves as a left-moving and a right-moving wave packet each propagating with momentum $\frac{8}{n+2}$. In Appendix \ref{app:linenum} we present numerical data that suggest that this type of evolution can be already observed for small values of the time $t$ and length $n$.

\section{Quantum snake walk on the glued trees graph} \label{sec:GTG}

In this section we discuss how the continuous-time quantum snake walk might potentially lead to an efficient algorithm solving the extended glued trees problem, which we defined in Section \ref{ssec:extdef}. Our aim is to come up with an algorithm which makes $O(\poly N)$ oracle queries, where $N$ is the height of the glued trees graph and, without loss of generality, we assume each vertex is encoded using $\Theta(N)$ bit label.

\subsection{The main idea} \label{ssec:algIdea}

Let us consider the following expansion of the glued trees graph. Suppose $G$ is a glued trees graph of height $N$, and $r_1$ and $r_2$ are its roots. For $M>N$ we construct an {\em expanded glued trees graph} $G^M$ as follows: take $2^{M-N}$ instances of graph $G$, create two complete binary trees $T_1$ and $T_2$ each of height $M-N$, and then make bijective associations between leaves of $T_1$ and instances of $r_1$ and between leaves of $T_2$ and instances of $r_2$ to obtain $G^M$. An example of the graph $G^4$ for a given glued trees graph $G$ of height $2$ is given in Figure \ref{fig:egtg}. Basically, $G^M$ looks exactly like a glued trees graph of height $M$, but instead of one cycle gluing the trees, now we have $2^{M-N}$ cycles.

For the rest of the paper we will focus on the continuous-time quantum snake walk on $G^M$. Given an oracle for $G$, since we know an efficient algorithm finding the label of $r_2$ \cite{CCDFGSart}, we can efficiently simulate the quantum snake walk on $G^M$ for $M\in O(\poly N)$ and the length of the snake $n\in O(\poly N)$ (see \cite[Section 4.1.3]{RosmanisThesis}). The sketch of an algorithm potentially solving the extended glued trees problem would be as follows:
\begin{enumerate}
  \item choose $M,n\in O(\poly N)$ such that $n\geq 2N+1$;
  \item \label{s:alg2} construct a particular initial superposition for a quantum snake walk on $G^M$ which contains only snakes on $T_1$ (construction of any such state involves no oracle queries at all);
  \item run a quantum snake walk on $G^M$ for a certain amount of time $t\in O(\poly N)$ and measure, thus obtaining a snake of length $n$;
  \item if the obtained snake contains two vertices such that one belongs to $T_1$ and other to $T_2$, then the snake has to go completely through one instance of $G$, and from this snake we can extract a path connecting $r_1$ and $r_2$; otherwise return to Step \ref{s:alg2}.
\end{enumerate}

\begin{figure}[htp]
\centering
\includegraphics[scale=0.9]{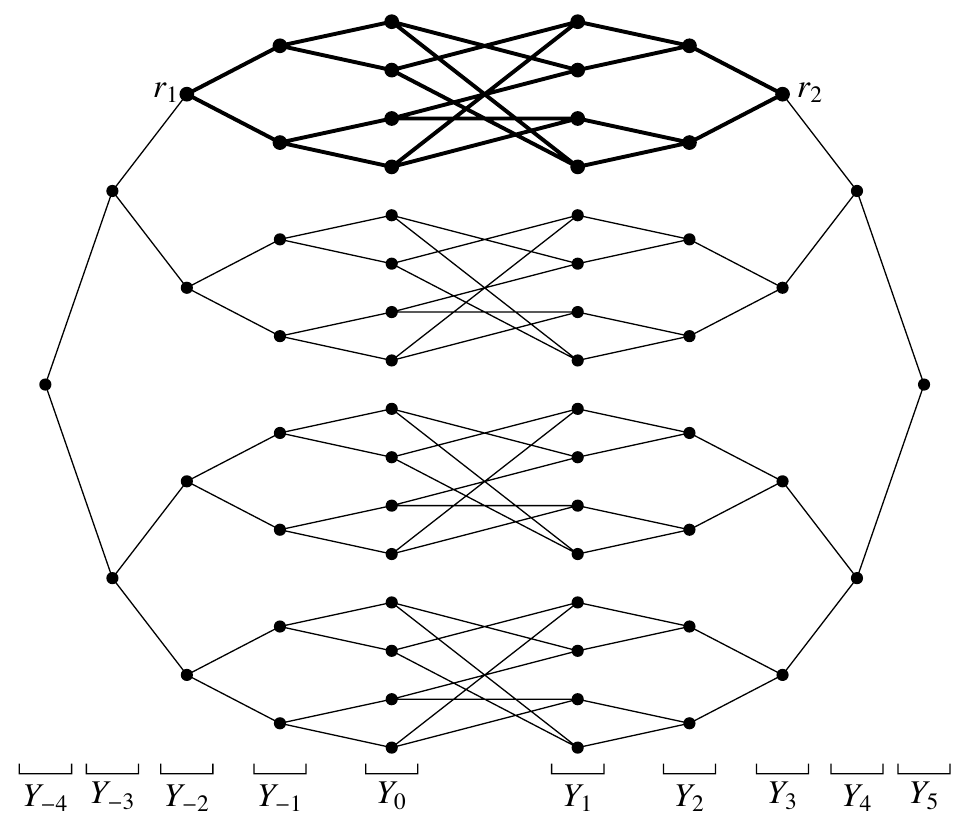}
\caption{An expanded glued trees graph $G^M$ for $N=2$ and $M=4$. One instance of the original glued trees graph $G$ is highlighted using thick edges.}
\label{fig:egtg}
\end{figure}

First let us show that the analysis of the quantum snake walk on this relatively complex graph becomes much easier if we restrict the choice of initial state to certain superpositions. For $x\in[-M\,..\,M+1]$ let $Y_x$ be the set of all vertices in $G^M$ which are at distance $M+x$ from the root of $T_1$ (see Figure \ref{fig:egtg}). Now, for $x\in[-M\,..\,M+1]$ and $j\in\{0,1\}^n$ let $S(x,j)$ be the set of all snakes $s=(v_0,\ldots,v_n)$ such that $v_0\in Y_x$ and for all $l\in[1\,..\,n]$, if $v_{l-1}\in Y_z$, then $v_{l}\in Y_{z-(-1)^{j_l}}$. For $S(x,j)\neq\emptyset$ let $|x,j\rangle$ be the uniform superposition over the elements of $S(x,j)$, and let us refer to $\spn\{|x,j\rangle\,:\,S(x,j)\neq\emptyset\}$ as the {\em column subspace}. Due to symmetry, the column subspace is invariant under $A_n(G^M)$, the Hamiltonian governing the quantum snake walk on $G^M$, and the action of $A_n(G^M)$ on the column subspace does not depend on which cycle in particular glues the two binary trees in $G$. Also, let us call the subgraph of $G^M$ having the set of vertices $Y_0\cup Y_1$ and all the edges originally connecting $Y_0$ and $Y_1$ the {\em glued part}.

From now on let us consider $M$ such that $M\gg n$. For $|x|\ll M$ and $j\in\{0,1\}^n$,  if we run the quantum snake walk on $G^M$ starting from the state $|x,j\rangle$ for time small enough, then we will never get close to the roots of $G^M$. Thus, in such a case the walk behaves almost exactly the same way on all $G^{M'}$ having $M'>M$, and for the sake of simpler analysis we assume $M=\infty$ (similarly as the regular quantum walk on a finite line segment, assuming the initial state is chosen to be far from the endpoints of the segment, can be approximated by the walk on the infinite line). Therefore the Hamiltonian governing the quantum snake walk on $G^M$ restricted to the column subspace effectively is
\begin{align}
   \tilde{H}_n=\!\!\sum_{j\in\{0,1\}^{n-1}}\!
   \Bigg(
   \sum_{x=-\infty}^{-\max(|j|_\pm,0)}
   &\Big(
   2\,(\,|x+1,0j\rangle\langle x,j1|+|x,j1\rangle\langle x+1,0j|\,) \notag
\\ &
   +  (\,|x-1,1j\rangle\langle x,j0|+|x,j0\rangle\langle x-1,1j|\,) \label{eqp:HHat}
\\ &
   + \sqrt{2}\,(\,|x+1,0j\rangle\langle x,j0|+|x,j0\rangle\langle x+1,0j|\,) \notag
\\ &
   + \sqrt{2}\,(\,|x-1,1j\rangle\langle x,j1|+|x,j1\rangle\langle x-1,1j|\,) \notag
   \Big)
\\
   + \sum_{x=1-\max(|j|_\pm,0)}^{-\min(|j|_\pm,0)}
   &\Big(
     \sqrt{2}\,(\,|x+1,0j\rangle\langle x,j1|+|x,j1\rangle\langle x+1,0j|\,) \notag
\\ &
   + \sqrt{2}\,(\,|x-1,1j\rangle\langle x,j0|+|x,j0\rangle\langle x-1,1j|\,)
\\ &
   + \text{If}_{|j|_\pm>0}[2;1]\,(\,|x+1,0j\rangle\langle x,j0|+|x,j0\rangle\langle x+1,0j|\,) \notag
\\ &
   + \text{If}_{|j|_\pm>0}[1;2]\,(\,|x-1,1j\rangle\langle x,j1|+|x,j1\rangle\langle x-1,1j|\,) \notag
   \Big)
\\
   + \sum_{x= 1-\min(|j|_\pm,0)}^{+\infty}
   &\Big(
        (\,|x+1,0j\rangle\langle x,j1|+|x,j1\rangle\langle x+1,0j|\,) \notag
\\ &
   + 2\,(\,|x-1,1j\rangle\langle x,j0|+|x,j0\rangle\langle x-1,1j|\,)
\\ &
   + \sqrt{2}\,(\,|x+1,0j\rangle\langle x,j0|+|x,j0\rangle\langle x+1,0j|\,) \notag
\\ &
   + \sqrt{2}\,(\,|x-1,1j\rangle\langle x,j1|+|x,j1\rangle\langle x-1,1j|\,) \notag
   \Big)
   \Bigg),   
\end{align}
where $|j|_\pm=-\sum_{y=1}^{n-1}(-1)^{j_y}$ and $\text{If}_{|j|_\pm>0}[a_1;a_2]$ is $a_1$ if $|j|_\pm>0$, and $a_2$ otherwise (see \cite[Section 4.2]{RosmanisThesis} for details).

\subsection{Quantum snake walk on the infinite binary tree}

Here let us consider the continuous-time quantum snake walk on $G^M$ restricted to superpositions $|x,j\rangle$ such that $-M\ll x\ll -n$. We care only about the first sum (\ref{eqp:HHat}) in the expression for $\tilde{H}_n$, and, because $x\ll-n$ , similarly as before, we can approximate an action of $\tilde{H}_n$ by one of
\begin{equation*}
\hat{H}_n = \int_0^{2\pi}{|\tilde{k}\rangle\langle\tilde{k}|\otimes \hat{H}_{n,k}\,\dd k},
\end{equation*}
where
\begin{equation*}
\begin{split} 
  \hat{H}_{n,k} \; = \; & \sum_{j\in\{0,1\}^{n-1}} {e^{\ii k}\;(\,\sqrt{2}|j0\rangle\langle 0j|+2|j1\rangle\langle 0j|+|1j\rangle\langle j0|+\sqrt{2}|1j\rangle\langle j1| \,)} \\
             +    & \sum_{j\in\{0,1\}^{n-1}} {e^{-\ii k} (\,\sqrt{2}|0j\rangle\langle j0|+2|0j\rangle\langle j1|+|j0\rangle\langle 1j|+\sqrt{2}|j1\rangle\langle 1j| \,)}.
\end{split}
\end{equation*}
The analysis of Hamiltonian $\hat{H}_n$ and  the quantum evolution governed by it will highly resemble the analysis of $H_n$ above. 
Let us redefine $|u_{0,k}\rangle=\frac{1}{\sqrt{3}}(\sqrt{2}e^{-\ii k}|0\rangle+e^{\ii k}|1\rangle)$, $|u_{1,k}\rangle=\frac{1}{\sqrt{3}}(e^{-\ii k}|0\rangle-\sqrt{2}e^{\ii k}|1\rangle)$, $|v_0\rangle=\frac{1}{\sqrt{3}}(|0\rangle+\sqrt{2}|1\rangle)$ and $|v_1\rangle=\frac{1}{\sqrt{3}}(\sqrt{2}|0\rangle-|1\rangle)$, and let us define the orthonormal basis $\{ |\widehat{0}_k\rangle, \ldots, |\widehat{2^n-1}_k\rangle \}$ as before except that now we are using these new definitions of $|u_{0,k}\rangle$, $|u_{1,k}\rangle$, $|v_0\rangle$ and $|v_1\rangle$ to define  $|\widehat{m}_k\rangle$ for each  $m\in[0\,..\,2^n-1]$. In this new basis
\begin{equation*}
\begin{split} 
  \hat{H}_{n,k}\;=\; 
  4\sqrt{2} \cos k \, |\widehat{0}_k\rangle\langle\widehat{0}_k|\;+ \;
   (3 \sin k +\ii \cos k)\,|\widehat{1}_k\rangle\langle\widehat{0}_k| & \;+\;
   (3 \sin k -\ii \cos k)\, |\widehat{0}_k\rangle\langle\widehat{1}_k| \\
   +  \;& 3 \sum_{m=1}^{2^{n-1}-1}{(|\widehat{2m}_k\rangle\langle\widehat{m}_k|+|\widehat{m}_k\rangle\langle\widehat{2m}_k|)}.
\end{split} 
\end{equation*}
Again we can see that only $n+1$ eigenvalues of $\hat{H}_{n,k}$ are $k$-dependent, and those are the eigenvalues of operator $\hat{\Phi}_{n,k}=\hat{U}^*_{n,k}\hat{H}_{n,k}\hat{U}_{n,k}$, where $\hat{U}_{n,k}=\sum_{y=1}^{n}{|\widehat{2^{n-y}}_k\rangle\langle\overline{y}|+|\widehat{0}_k\rangle\langle\overline{n+1}|}$. 
Along similar lines as for Lemma \ref{the:pcondition} we can prove the following lemma.

\begin{lemma}
Let us fix $n\in\N$ and $k\in\R$. The equation
\begin{equation*}
  6(3\cos p -2\sqrt{2}\cos k)\sin ((n+1)p) =(1+8 \sin^2 k) \sin np
\end{equation*}
has $n+1$ distinct solutions in the interval $(0,\pi)$, and, if $p$ is  a solution of this equation, then $6\cos p$ is an eigenvalue of $\hat{\Phi}_{n,k}$.
\end{lemma}

Now suppose $n$ is even, and let $\tilde{\lambda}(k)$ be the median (i.e., the $\frac{n+2}{2}$-th largest) eigenvalue of $\hat{\Phi}_{n,k}$. Let $\tilde{\Lambda}(k)=6\arctan\left(\frac{12\sqrt{2}\cos k}{1+8\sin^2k}\right)$. Similarly as before, we have

\begin{lemma}
For every $n$, $\tilde{\lambda}'(k)$ is bounded between $\frac{\tilde{\Lambda}'(k)}{n}\left(1-\frac{2}{n}\right)$ and $\frac{\tilde{\Lambda}'(k)}{n}\left(1+\frac{2}{n}\right)$ and we have  $\left|\tilde{\lambda}''(k)-\frac{\tilde{\Lambda}''(k)}{n}\right|\in O\left(\frac{1}{n^2}\right)$. Also, there exists $\tilde{n}_0\in\N$ such that for all $n\geq \tilde{n}_0$ we have $\tilde{\lambda}''(k)=0$ if and only if $k\equiv\frac{\pi}{2}\mod\pi$.
\end{lemma}

Let $|\hat{\phi}(k)\rangle\in\C^{n+1}$ be the unique eigenvector of $\hat{\Phi}_{n,k}$ corresponding to the eigenvalue $\tilde{\lambda}(k)$ such that $\langle\hat{\phi}(k)|\hat{\phi}(k)\rangle=1$ and $\langle\overline{1}|\hat{\phi}(k)\rangle>0$, let $|\hat{\psi}(k)\rangle=\hat{U}_{n,k}|\hat{\phi}(k)\rangle$, and let 
\begin{equation*}
|\hat{\eta}_x\rangle = \frac{1}{\sqrt{2\pi}}
\int_{0}^{2\pi}
e^{-\ii k x}|\tilde{k}\rangle\otimes|\hat{\psi}(k)\rangle\,\dd k.
\end{equation*}
Similarly as for the quantum snake walk on the line,  numerical results suggest that $|\hat{\eta}_x\rangle$ is localized on the tree $T_1$ at the distance $|x|$ away from the glued part. If this property indeed holds, we can show that for $n\geq \tilde{n}_0$ and for asymptotically large time of the evolution $t$ the state $|\hat{\eta}_x\rangle$ evolves as two wave packets moving with momentum $\frac{8\sqrt{2}}{n+2}$ in the opposite directions, that is, one moving up the tree and the other down the tree. However, now we need to be careful when we talk about asymptotically large $t$, because we still want $t$ to be small enough so that the wave packets do not get close to the root of $T_1$ or the glued part.

For the purpose of our algorithm, we are interested in only one of those two wave packets, the one which moves towards the glued part. Our hope is that in $O(\poly N)$ time this wave packet will propagate far enough to reach $T_2$ and with high probability will give us a snake containing vertices from both $T_1$ and $T_2$. The regular quantum walk on the line started from a single vertex evolves as two wave packets moving in the opposite directions with momentum $2$ \cite[Section 3.3.2]{ChildsThesis}, but for any $\omega\in[0,2]$ we can create a wave packet which moves with momentum $\omega$ in one single direction (see \cite{ChildsUniv}). We would like to do an analogous thing for the quantum snake walk. For $k_0\in(\pi,2\pi)$ and $\sigma>0$ consider the state
\begin{equation*}
\begin{split}
|\xi_{x_0,k_0,\sigma}\rangle & = \frac{1}{\sqrt{\erf(\pi/\sqrt{2}\sigma)}}
\int_{k_0-\pi}^{k_0+\pi}
\sqrt{\frac{1}{\sigma\sqrt{2\pi}}e^{-\frac{(k-k_0)^2}{2\sigma^2}}}
e^{-\ii kx_0}|\tilde{k}\rangle\otimes|\hat{\psi}(k)\rangle\,\dd k \\
& =
\sum_{z\in\Z}
e^{\ii k_0z}\sqrt{\frac { e^{\frac{-z^2}{2(1/2\sigma)^2}}} {(1/2\sigma)\sqrt{2\pi}}}
\frac{\Re\left(\erf\left(\frac{\pi}{2\sigma}+\ii z\sigma\right)\right)}{\sqrt{\erf(\pi/\sqrt{2}\sigma)}} \; |\hat{\eta}_{x_0+z}\rangle \\
& \approx
\sum_{z=-\lfloor\frac{1}{2\sigma^2}\rfloor}^{\lfloor\frac{1}{2\sigma^2}\rfloor}e^{\ii k_0z} \sqrt{\binom{2\lfloor\frac{1}{2\sigma^2}\rfloor}{z+\lfloor\frac{1}{2\sigma^2}\rfloor} \Big/ 2^{2\lfloor\frac{1}{2\sigma^2}\rfloor} }
\;|\hat{\eta}_{x_0+z}\rangle,
\end{split}
\end{equation*}
where $\erf(\cdot)$ denotes the error function and $\Re(\cdot)$ denotes the real part function. One can show that $\langle\xi_{x_0,k_0,\sigma}|\xi_{x_0,k_0,\sigma}\rangle=1$. Numeric results suggest that $|\xi_{x_0,k_0,\sigma}\rangle$ evolves as a Gaussian-shaped wave packet moving towards the glued part of $G^M$ with momentum $\tilde{\lambda}'(k_0)$ (see Appendix \ref{app:GMnum}).

For any snake in the superposition $e^{-\ii \hat{H}_n t}|\xi_{x_0,k_0,\sigma}\rangle$, other than the position of its initial vertex, we are also interested in how many different depths in the tree $T_1$ its vertices possesses. To be more precise, for a snake $s=(v_0,\ldots,v_n)$ let $(z_0,\ldots,z_n)\in\Z^{n+1}$ be such that $v_l\in Y_{z_l}$ for all $l\in[0\,..\,n]$. We define the {\em span} of $s$ to be $(\min_lz_l,\max_lz_l)$ and the {\em span length} of $s$ to be $\max_lz_l-\min_lz_l$. We are interested in snakes with reasonably large span length, because, after all, if we obtain a snake which contains a solution of the extended glued trees problem, the span length of this snake will be at least $2N+1$. Numerics presented in Appendix \ref{app:GMnum} suggest that, if we measure $e^{-\ii \hat{H}_n t}|\xi_{x_0,k_0,\sigma}\rangle$ in the standard basis and output the span length of the obtained snake, then the expected value of this random variable is in $\Omega(\sqrt{n})$, which is good. However, the wave packet might drastically loose its average span length when it gets close to the glued part and the action of the operator $\tilde{H}_n$ cannot be approximated anymore by one of $\hat{H}_n$. In Appendix \ref{app:unfort} we describe numerical data which, unfortunately, suggest that it indeed may be the case.

\subsection{Scattering on the glued part}

Let us consider yet another analogy to the regular quantum walk on lines. Consider the regular quantum walk on a graph $A$ constructed from some finite graph $B$ by attaching two semi-infinite lines to two vertices of $B$. On either of those two lines we can construct a wave packet which moves towards the graph $B$. What are the probabilities of this wave packet being reflected form $B$ and it being transmitted through $B$ (to the other line) can be calculated by inspecting the eigenvalues and eigenvectors of the adjacency matrix of $A$ and then using the standard scattering theory (see \cite{ChildsUniv}, for example).

Similarly, for the quantum snake walk on $G^M$ we can construct wave packets which move on $T_1$ towards the glued part (to be more precise, we seem to be able to construct such wave packets, since we have not rigorously proven that yet), we can also construct wave packets which move on $T_1$ away from the glued part, and we can do the same on $T_2$.
In order to understand what happens to the wave packets moving towards the glued part once they reach it, let us inspect particular eigenvectors and eigenvalues of $\tilde{H}_n$. 
Let $|\check{\psi}(k)\rangle$ be the unique vector satisfying $\langle j|\check{\psi}(k)\rangle=\langle \hat{\psi}(k)|j\oplus1^n\rangle$ for all $j\in\{0,1\}^n$. Numerical data for $n$ up to $10$ agree with the following hypothesis.
\begin{hypothesis}
\label{hyp:evec}
For every $k\in\R$ $\tilde{\lambda}(k)$ is the eigenvalue of $\tilde{H}_n$ coresponding to the eigenvector
\begin{equation*}
\sum_{x=-\infty}^{-n+1} \!\!\! |x\rangle\otimes\big(e^{\ii k x}|\hat{\psi}(k)\rangle + R(k)\,e^{-\ii k x}|\hat{\psi}(-k)\rangle\big) + \!\!
\sum_{x=-n+2}^{n-1} \!\!\!\!  |x\rangle\otimes \!\!\!\!\! \sum_{j\in\{0,1\}^n} \!\!\!\! \alpha_{x,j}(k)|j\rangle +
\sum_{x=n}^{+\infty} \!  |x\rangle\otimes T(k)\, e^{\ii k x}|\check{\psi}(k)\rangle,
\end{equation*}
where $\alpha_{x,j}(k)\in\C$,
\begin{equation*}
R(k)=\frac{1-2e^{2\ii k}}{5-2e^{-2\ii k}-2e^{2\ii k}} \quad\text{and}\quad T(k)=\frac{\sqrt{2}(e^{-2\ii k}-3+2e^{2\ii k})}{5-2e^{-2\ii k}-2e^{2\ii k}}.
\end{equation*}
\end{hypothesis}
Let us denote the eigenvector given in Hypothesis \ref{hyp:evec} by  $|\mu(k)\rangle$. 
Assuming Hypothesis \ref{hyp:evec}, the standard scattering theory suggests that the wave packet (started as) $|\xi_{x_0,k_0,\sigma}\rangle$, where $k_0\in(\pi,2\pi)$ and $\sigma$ is assumed to be small, will be reflected from the glued part with probability $|R(k_0)|^2=\frac{1}{1+8\sin^2\!k_0}$ and transmitted through it with the probability $|T(k_0)|^2=\frac{8\sin^2\!k_0}{1+8\sin^2\!k_0}$. Also, $\arg'(T(k_0))$, called the {\em effective length}, divided my the momentum of the wave packet describes for how much time the glued part delays the wave packet propagating through it. It is really interesting that the transmission coefficient $T(k)$ does not depend on $n$, the length of the snake, and it still stays the same if we consider eigenvectors of $\tilde{H}_n$ corresponding to any other $k$-dependent eigenvalue of $\hat{H}_{n,k}$, not only the median eigenvalue. Even more, already in \cite{CCDFGSart} Childs et al.\ proved that for the regular quantum walk on $G^M$ the probability of a wave packet corresponding to the vector $|\tilde{k}\rangle$ being transmitted through the glued part is $\frac{8\sin^2\!k}{1+8\sin^2\!k}$. Figure \ref{fig:transm} shows that the wave packet $|\xi_{x_0,k_0,\sigma}\rangle$ has the highest probability of being transmitted through the glued part and has the minimal effective length of the glued part for $k_0=\frac{3\pi}{2}$. It also seems that $\tilde{\lambda}'(k_0)$ reaches its maximum at $k_0=\frac{3\pi}{2}$ (we have $\tilde{\lambda}'(\frac{3\pi}{2})=\frac{8\sqrt{2}}{n+2}$), therefore the faster a wave packet is, the most likely it will propagate through the glued part (note: this is the case only for wave packets corresponding to the median eigenvalue).

\begin{figure}[htp]
\centering
\includegraphics[scale=0.95]{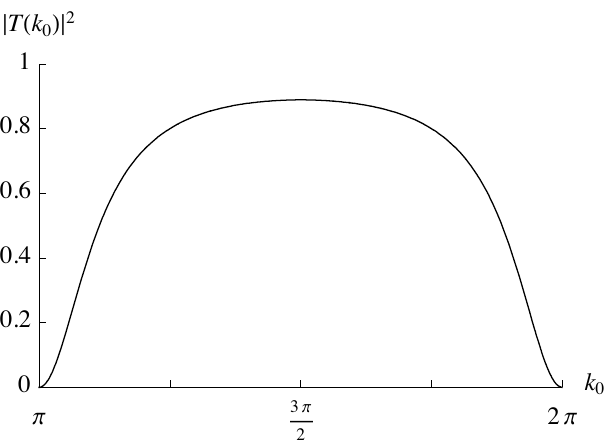}
\quad
\includegraphics[scale=0.95]{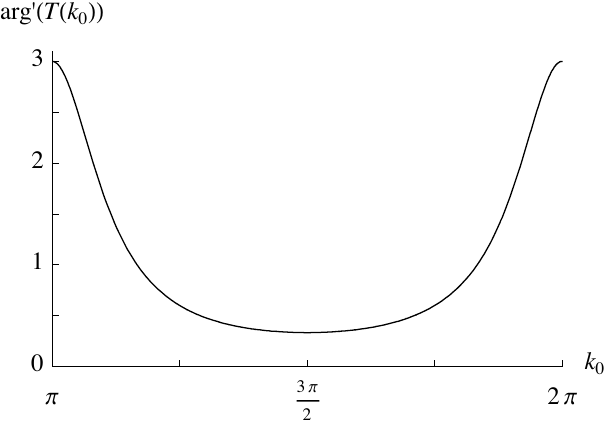}
\caption{The likely probability of the wave packet $|\xi_{x_0,k_0,\sigma}\rangle$ propagating through the glued part on the left, and, on the right, its corresponding effective length of the glued part.}
\label{fig:transm}
\end{figure}

Even if we rigorously proved that we can construct localized wave packets having momenta as stated above, the span lengths of those wave packets are as stated, and so are the transmission probabilities, such a proof most likely would use the method of stationary phase approximation, which deals with the case when all (i.e., initial, reflected and transmitted) wave packets are far away from the glued part. But the problem is: what happens during the time when wave packets encounter the glued part is exactly what we are interested in for the analysis of our algorithm. Unfortunately we cannot obtain enough numerical data to have a good intuition what is the maximal probability obtaining a snake connecting $T_1$ and $T_2$ for any given $N$ and $n\geq 2N+1$. Also, $n$ necessary for an efficient algorithm may be very large because we only restrict it to be in $O(\poly N)$.
Nonetheless, a numerical inspection of the vector $|\mu(k)\rangle$ done in Appendix \ref{app:unfort} suggests that the expected span length of snakes in a wave packet is significantly reduced when the wave packet propagates through the glued part. We currently do not have any candidates for a pattern this reduction follows, so it is not clear if the reduction is exponential or not.

\section{Discussion}

In this paper we introduced the continuous-time quantum snake walk and discussed one of its potential algorithmic applications, the extended glued trees problem. Assuming the algorithm sketched in Section \ref{ssec:algIdea} which uses wave packets as initial states indeed efficiently finds a path connecting both roots of the glued trees graph, it almost certainly will not find a shortest path. It would be interesting to see if there exists an efficient algorithm finding a shortest path, or an efficient quantum algorithm for any other pathfinding problem, which is based on the quantum snake walk.

On the other hand, if the algorithm we propose is not efficient, despite that the wave packets (most likely) are being transmitted through the glued part with high probability and in polynomial time, it would suggest that possibly there are some fundamental limitations which restrict us from solving the extended glued trees problem efficiently. Finding and understanding such limitations would be an interesting and useful thing to do.

Another potential direction for future research is to consider discrete-time analogues of the continuous-time quantum snake walk (for example, by generalizing Mc\,Gettrick's results \cite{McG}) and their applications.

\section*{Acknowledgments}
I wish to thank John Watrous, Richard Cleve, and Andrew Childs for numerous helpful discussions regarding this work. In particular, the idea to consider quantum walks whose states are paths is due to Watrous. This work was supported by a Mike and Ophelia Lazaridis Fellowship.

\appendix
\section{Proof of Lemma \ref{the:pcondition} } \label{app:pcondition}

Let $|\phi\rangle=\sum_{y=1}^{n+1}{a_y|\overline{y}\rangle}$ be an eigenvector of $\Phi_{n,k}$ and $\lambda$ be its corresponding eigenvalue. The maximum absolute column sum norm of (\ref{eq:lineHnk}) is at most $4$ (it is exactly $4$ unless $n=1$), which implies that all the eigenvalues of $H_{n,k}$ are at most $4$ by the absolute value, and therefore so are the eigenvalues of $\Phi_{n,k}$. In fact, it is easy to see that $|\lambda|=4$ would imply $k\equiv 0\mod\pi$ and $|\phi\rangle\propto|\overline{n+1}\rangle$, which is also the only case when $a_1=0$. Hence, there exist unique $p\in(0,\pi)$ and $c\neq0$ such that $\lambda=4\cos p$ and $a_1=c\sin p$.

We have
\begin{equation*}
  \lambda a_1 = \lambda \langle\overline{1}|\phi\rangle = \langle\overline{1}|\Phi_{n,k}|\phi\rangle = 2 \langle\overline{2}|\phi\rangle = 2 a_2,
\end{equation*}
from which we get $a_2=(4 \cos p\cdot c\sin p)/2=c\sin 2p$. Now let us use induction. Let $2\leq l<n$, and let us assume $a_y=c\sin yp$ for all $y\in[1\,..\,l]$.  We have
\begin{equation*}
  \lambda a_l = \lambda \langle \overline{l}|\phi\rangle = \langle \overline{l}|\Phi_{n,k}|\phi\rangle
  = (2 \langle \overline{l-1}| + 2 \langle \overline{l+1}|)|\phi\rangle = 2 a_{l-1}+2 a_{l+1},
\end{equation*}
which gives us
\begin{equation*}
  a_{l+1} = (4\cos p\cdot c\sin lp - 2c\sin (l-1)p)/2 = c \sin (l+1) p.
\end{equation*}
Therefore, by induction, $a_y=c\sin yp$ for $y\in[1\,..\,n]$. Similarly,
\begin{equation*}
  \lambda a_n = \lambda \langle \overline{n}|\phi\rangle = \langle \overline{n}|\Phi_{n,k}|\phi\rangle
  = (2 \langle \overline{n-1}| + 2\sin k\, \langle \overline{n+1}|)|\phi\rangle = 2 a_{n-1}+2a_{n+1}\sin k
\end{equation*}
implies $a_{n+1} = c \frac{\sin (n+1)p}{\sin k}$.
Finally,
\begin{equation*}
  \lambda a_{n+1} = \lambda \langle \overline{n+1}|\phi\rangle = \langle \overline{n+1}|\Phi_{n,k}|\phi\rangle
  = (2\sin k \langle \overline{n}| + 4\cos k \langle \overline{n+1}|)|\phi\rangle
   = 2\sin k\cdot a_n + 4 \cos k\cdot a_{n+1}
\end{equation*}
gives us
\begin{equation} \label{eq:pcondition2}
  2(\cos p -\cos k)\sin ((n+1)p) = \sin^2 k \sin np.
\end{equation}

One can see that (\ref{eq:pcondition2}) is both the necessary and sufficient condition for $4\cos p$ to be an eigenvalue of $\Phi_{n,k}$. Also, as shown above, an eigenvalue of $\Phi_{n,k}$ uniquely (up to a global factor) determines the eigenvector corresponding to it. This implies that all the eigenvalues of $\Phi_{n,k}$ are distinct, thus (\ref{eq:pcondition2}) has $n+1$ distinct solutions in the interval $(0,\pi)$.

\section{Proof of Lemma \ref{lem:lambdaprime} } \label{app:hyplem}

For $n=2$ we have $\lambda(k)=2\cos(k)$ and the lemma holds, thus let us assume $n\geq 4$. 
We know that $\lambda(0)=4\cos\frac{\frac{n}{2}\pi}{n+1}$ and $\lambda(\pi)=4\cos\frac{(\frac{n}{2}+1)\pi}{n+1}$, therefore a corollary of Theorem \ref{thm:lambdaD} is that $|\lambda(k)|\leq4\sin\frac{\pi}{2n+2}$. Let $p(k)=\frac{\pi}{2}-\frac{\theta(k)}{n+1}\in[\frac{n\pi}{2n+2},\frac{(n+2)\pi}{2n+2}]$ be such that $\lambda(k)=4\cos p(k)$, and therefore $\theta(k)\in[-\frac{\pi}{2},\frac{\pi}{2}]$. Hence $\lambda(k)=4\sin(\frac{\theta(k)}{n+1})$, where $\theta(k)$ is the solution of
\begin{equation*}
2\left(\cos k-\sin\left(\frac{1}{n+1}\theta(k)\right)\right)\cos\theta(k) = \sin^2 k \sin \left(\frac{n}{n+1}\theta(k)\right),
\end{equation*}
which is just a rewritten $p$-equation.
By simple trigonometric derivations this gives us
\begin{equation*}
\tan\theta(k)=
\frac{2\frac{\cos k-\sin\left(\frac{\theta(k)}{n+1}\right)}{\sin^2 k}+\sin \left(\frac{\theta(k)}{n+1}\right)}{\cos \left(\frac{\theta(k)}{n+1}\right)}.
\end{equation*}
Notice that $(\tan\theta(k))'=\theta'(k)(1+\tan^2\theta(k))$, and we already have an expression for $\tan\theta(k)$. Thus we can obtain an expression for $\theta'(k)$, from which 
\begin{equation*}
\lambda'(k)
=\frac  {4\theta'(k)\cos\frac{\theta(k)}{n+1}}  {n+1}
=-\frac {16\sin k\cdot\cos^2\!\frac{\theta(k)}{n+1}}  {n(3+\cos{2k})+4-4(n+1)\cos k\cdot\sin\frac{\theta(k)}{n+1}}.
\end{equation*}
Because $\lambda(\kappa)=\lambda(-\kappa)$ and $\lambda(\frac{\pi}{2}+\kappa)=-\lambda(\frac{\pi}{2}-\kappa)$ for any $\kappa\in\R$, without the loss of generality let us assume $k\in(0,\frac{\pi}{2})$; hence $\theta(k)\in(0,\frac{\pi}{2})$. Therefore $0\leq\sin\frac{\theta(k)}{n+1}\leq\frac{\pi}{2(n+1)}$,  $1-\frac{\pi^2}{8(n+1)^2}\leq\cos\frac{\theta(k)}{n+1}\leq1$ and
\begin{equation*}
0   \leq   4(n+1)\cos k\cdot\sin\frac{\theta(k)}{n+1}   \leq   2\pi\cos k   \leq   2\pi,
\end{equation*}
which implies
\begin{equation*}
\begin{split}
\lambda'(k)
& \geq -\frac{16\sin k}{n(3+\cos{2k})-2(\pi-2)} \geq -\frac{16\sin k}{n(3+\cos{2k})}\left(1+\frac{2}{n}\right), \\
\lambda'(k)
& \leq -\frac{16\sin k\cdot\left(1-\frac{\pi^2}{8(n+1)^2}\right)^2}{n(3+\cos{2k})+4}
\leq -\frac{16\sin k}{n(3+\cos{2k})}\left(1-\frac{2}{n}\right).
\end{split}
\end{equation*}
Since we already have expressions for $\lambda'(k)$ and $\theta'(k)$, we can get that
\begin{equation}
\begin{split}
\lambda''(k)=8 \cos^2\frac{\theta(k)}{n+1}
\Bigg(
& -
\frac{n (\cos k+\cos^3\!k+2 \cos k \sin^2\!k)}
{\left(n(1+\cos^2\!k)+2-2(n+1)\cos k\cdot\sin\frac{\theta(k)}{n+1}\right)^2} \\
& +
\frac{2 (n-1+2 \cos^2\!k)\sin\frac{\theta(k)}{n+1}-2 \cos k}
{\left(n(1+\cos^2\!k)+2-2(n+1)\cos k\cdot\sin\frac{\theta(k)}{n+1}\right)^2} \\
& +
\frac{4 (n+1) \cos k \sin^2\!k  \cdot\cos^2\frac{\theta(k)}{n+1}}
{\left(n(1+\cos^2\!k)+2-2(n+1)\cos k\cdot\sin\frac{\theta(k)}{n+1}\right)^3}
\Bigg).
\end{split}
\end{equation}
The first term of this equality is in $\Theta\left(\frac{1}{n}\right)$ while the other two are in $O\left(\frac{1}{n^2}\right)$. More precise analysis of the first term reveals that there is a constant $a>0$ such that $\left|\lambda''(k)-\frac{\Lambda''(k)}{n}\right|\leq \frac{a}{n^2}$.

Regarding the second part of the lemma, along the similar lines as for $\lambda''(k)$, one can show that there also exists a constant $b>0$ such that $\left|\lambda^{(3)}(k)-\frac{\Lambda^{(3)}(k)}{n}\right|\leq \frac{b}{n^2}$. Because of that, there exist $n_0\in\N$ and $\kappa\in(0,\frac{\pi}{2})$ such that $\Lambda''(k)+\frac{a}{n_0}<0$ for all $k\in[0,\kappa]$ and $\Lambda^{(3)}(k)-\frac{b}{n_0}>0$ for all $k\in[\kappa,\frac{\pi}{2}]$. Consider $n\geq n_0$. We have $\lambda''(k)<0$ for all $k\in[0,\kappa]$. The fact that $\lambda''(\frac{\pi}{2})=0$ and $\lambda^{(3)}(k)>0$ for all $k\in[\kappa,\frac{\pi}{2}]$ implies $\lambda''(k)<0$ for $k\in[\kappa,\frac{\pi}{2})$. Thus, $\lambda''(k)\neq0$ for all $n\geq n_0$ and $k\in[0,\frac{\pi}{2})$.

\section{Numerics for the quantum snake walk on the line} \label{app:linenum}

Fix even $n\in\Z$. For any given $k$ we can numerically calculate the median eigenvalue $\lambda(k)$ of $\Phi_{n,k}$ and its corresponding eigenvector $|\phi(k)\rangle$. Then we can get a numerical approximation of $|\eta_0\rangle$ by approximating the integral (\ref{eq:defeta}) defining it by a Riemann sum.

For every unit vector $|\chi\rangle\in\C^{\Z\times\{0,1\}^n}$, $y\in\Z$ and $j\in\{0,1\}^n$ let $\chi_{y,j}=\langle y,j|\chi\rangle$ and let $|\chi\rangle_x=\sum_{j\in\{0,1\}^n}\chi_{x,j}|x,j\rangle$. Then, in the standard basis, we have $||\chi\rangle_x|_1=\sum_{j\in\{0,1\}^n}|\chi_{x,j}|$ and  $||\chi\rangle_x|_2^2=\sum_{j\in\{0,1\}^n}|\chi_{x,j}|^2$. If we measure $|\chi\rangle$ in the standard basis, $||\chi\rangle_x|_2^2$ is the probability of obtaining a snake which starts at the position $x$; therefore the probability distribution $\{||\chi\rangle_x|_2^2\,:\,x\in\Z\}$ gives us good intuition about the `location' of the state $|\chi\rangle$ on the line.

\begin{lemma} \label{lem:eta0even}
For every even $x\in\Z$ we have $|\eta_0\rangle_x=0$.
\end{lemma}
\begin{proof}
Let $x$ be even, so that $e^{\ii k x}=e^{\ii (k+\pi) x}$. From  the definitions of $|\eta_0\rangle$ and $|\tilde{k}\rangle$ we can see that
\begin{equation*}
|\eta_0\rangle_x=\frac{1}{2\pi}|x\rangle\otimes\int_{0}^{2\pi}\!e^{\ii k x}\,|\psi(k)\rangle\,\dd k=\frac{1}{2\pi}|x\rangle\otimes\int_{0}^{\pi}\!e^{\ii k x}\,(\,|\psi(k)\rangle+|\psi(k+\pi)\rangle\,)\,\dd k.
\end{equation*}
We are left to show that $|\psi(k)\rangle=-|\psi(k+\pi)\rangle$ for all $k\in[0,\pi]$. For $y\in[1\,..\,n+1]$ let $\alpha_y\in\R$ be such that $|\phi(k)\rangle=\sum_{y=1}^{n+1}\alpha_y|\overline{y}\rangle$. One can see that $|\phi(k+\pi)\rangle=\sum_{y=1}^{n}(-1)^{y-1}\alpha_y|\overline{y}\rangle-\alpha_{n+1}|\overline{n+1}\rangle$. Now, $|u_{0,k}\rangle=-|u_{0,k+\pi}\rangle$ and $|u_{1,k}\rangle=-|u_{1,k+\pi}\rangle$ imply $|\widehat{2^{n-y}}_k\rangle=(-1)^y|\widehat{2^{n-y}}_{k+\pi}\rangle$ for $y\in[1\,..\,n]$ and $|\widehat{0}_k\rangle=|\widehat{0}_{k+\pi}\rangle$. From the definition of $U_{n,k}$ we see that $|\psi(k)\rangle=U_{n,k}|\phi(k)\rangle=-U_{n,k+\pi}|\phi(k+\pi)\rangle=-|\psi(k+\pi)\rangle$.
\end{proof}

\begin{figure}[htp]
\centering
\includegraphics[scale=0.85]{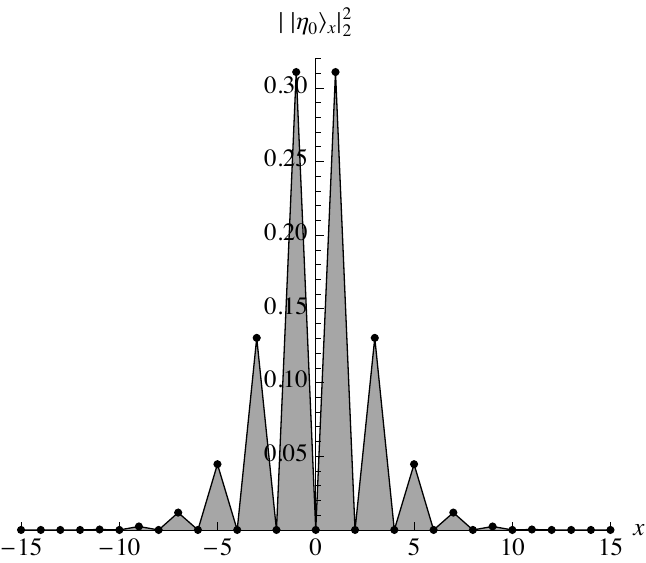}
\quad
\includegraphics[scale=0.85]{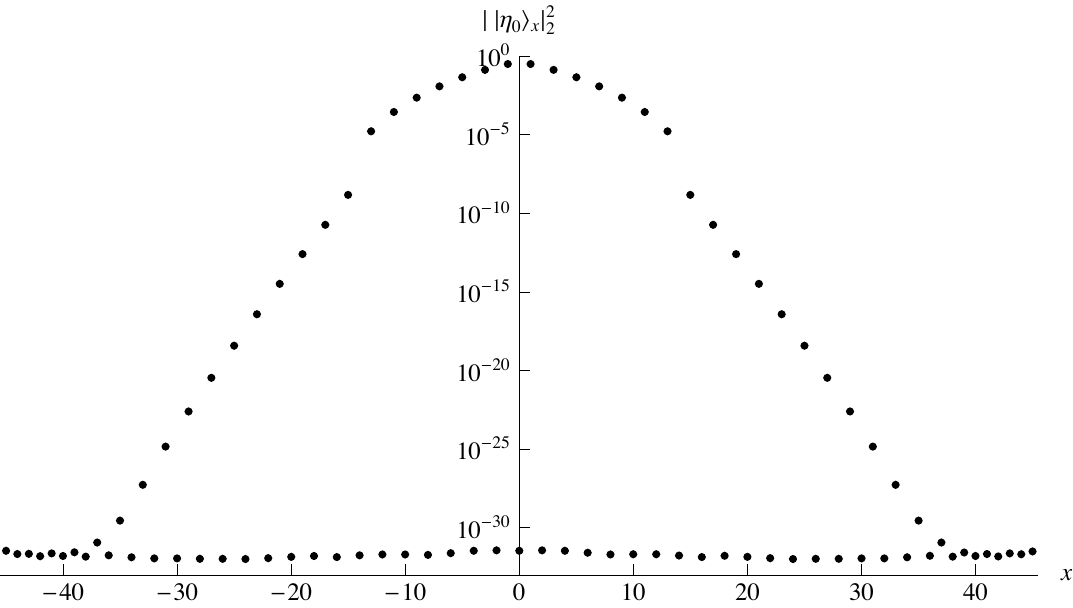}
\caption{For $n=14$, the probability $||\eta_0\rangle_x|_2^2$ of obtaining a snake starting at the position $x$ by measuring $|\eta_0\rangle$ in the standard basis (the linearly scaled plot on the left and the logarithmically scaled plot on the right). Note: numeric values below $10^{-30}$ cannot be trusted because they approach the precision of the floating point arithmetic used.}
\label{fig:eta0}
\end{figure}
 
As an evidence that $|\eta_0\rangle$ is localized around the position $0$ on the line, in Figure \ref{fig:eta0} we present numerically computed values of the probability $||\eta_0\rangle_x|_2^2$ in the case when $n=14$ for $x$ close to $0$. Lemma \ref{lem:eta0even} states that this probability is $0$ for even $x$, therefore let us now focus on odd $x$. We can see that $||\eta_0\rangle_x|_2^2$  seem to decrease exponentially as $|x|$ increases, and especially large jump seem to be from $||\eta_0\rangle_{14-1}|_2^2\approx1.61\cdot10^{-5}$ to $||\eta_0\rangle_{14+1}|_2^2\approx1.47\cdot10^{-9}$.

A similar pattern, a large gap between $\log(||\eta_0\rangle_{n-1}|_2^2)$ and $\log(||\eta_0\rangle_{n+1}|_2^2)$, was observed for values of $n$ other than $14$. This is the main motivation why Hypothesis \ref{hyp:local} is stated in the way it is stated. The reason why we choose $1$-norm in it instead of $2$-norm is a technicality: it allows us to make stronger statements about superpositions of vectors $|\eta_x\rangle$, where $x\in\Z$. Numerical data showed in Figure \ref{fig:tail} support Hypothesis \ref{hyp:local}.

\begin{figure}[htp]
\centering
\includegraphics[scale=0.95]{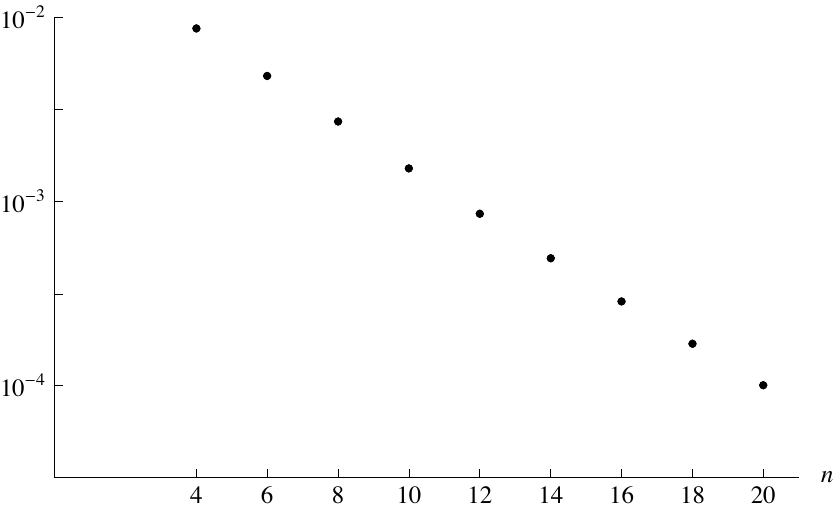}
\caption{In the logarithmic scale, $1$-norm of the part of the state $|\eta_0\rangle$ corresponding to snakes starting at least $n$ positions away from $0$ (note: $|\eta_0\rangle$ is $n$-dependent itself), that is, $\sum_{x=-\infty}^{-n}||\eta_0\rangle_x|_1+\sum_{x=n}^{+\infty}||\eta_0\rangle_x|_1$. For $n=2$ this value can be shown to be $0$.}
\label{fig:tail}
\end{figure}

Again, let $n=14$, and let us compute how $|\eta_0\rangle$ behaves under evolution of $H_{14}$. Section \ref{ssec:linewave} states that, under certain assumptions, for asymptotically large time $t$ the state  $|\eta_0\rangle$ evolves as two wave packets each moving with momentum $\frac{8}{n+2}=\frac{1}{2}$. We can see in Figure \ref{fig:2waves} that this kind of motion can be observed already for smaller values of $t$. Notice that peeks of the probability distribution around positions $\pm\frac{1}{2}t$ get relatively higher for larger $t$. Higher probabilities around the origin $x=0$ are likely there due to the fact that $\frac{1}{\sqrt{|\lambda''(k)|}}$ has a local maximum when $\lambda'(k)=0$.

\begin{figure}[htp]
\centering
\includegraphics[scale=0.95]{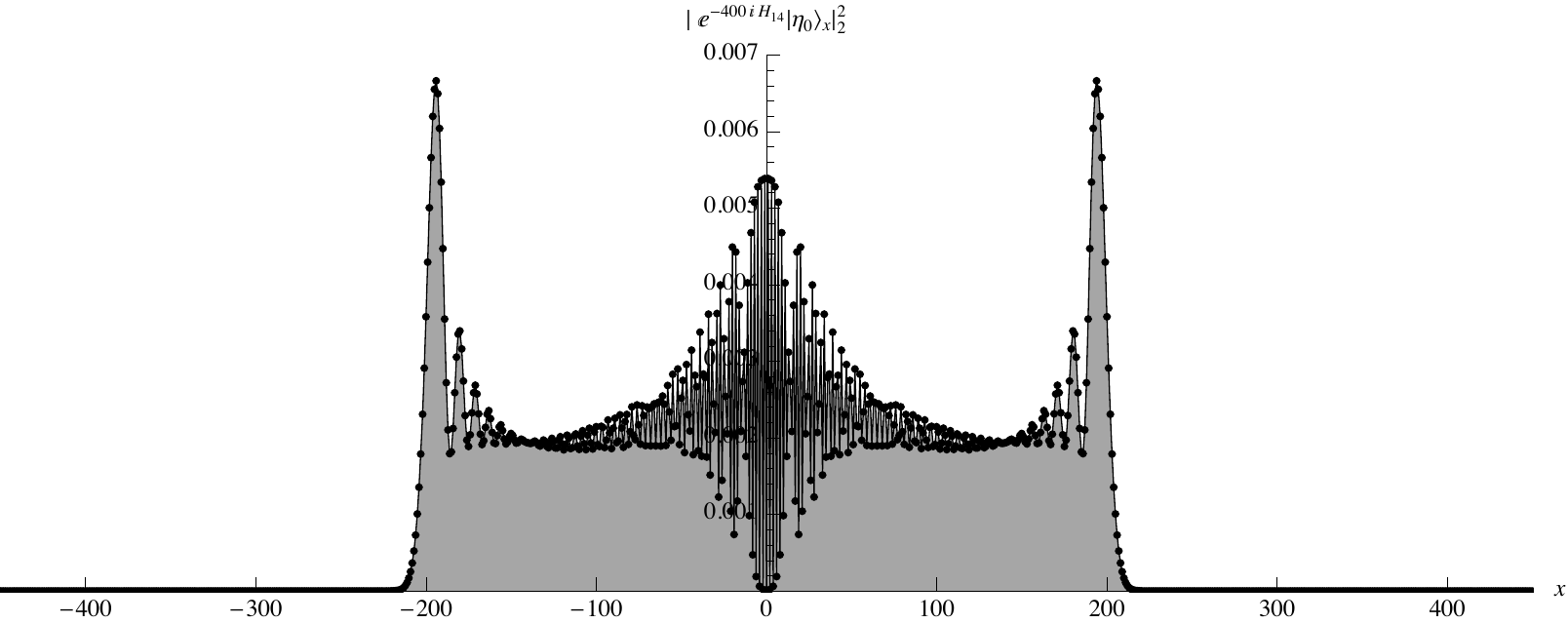}
\includegraphics[scale=0.95]{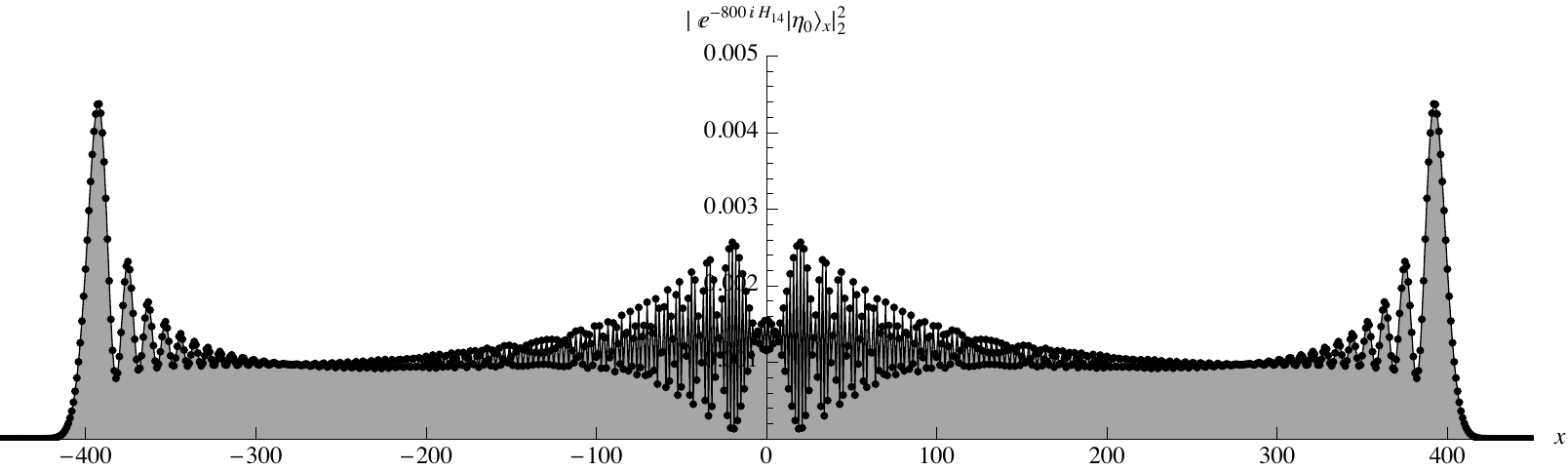}
\caption{The initial state $|\eta_0\rangle$ evolved for time $t=400$ and $t=800$.}
\label{fig:2waves}
\end{figure}

\section{Numerics for the quantum snake walk on the glued trees graph} \label{app:GMnum}

The same way as in the case of the quantum snake walk on the line, we can express $e^{-\ii \hat{H}_n t}|\xi_{x_0,k_0,\sigma}\rangle$ as an integral and then approximate it by a Riemann sum. Also, for every unit vector $|\chi\rangle\in\C^{\Z\times\{0,1\}^n}$ ($\C^{\Z\times\{0,1\}^n}$ being the column subspace) and $x\in\Z$ let $||\chi\rangle_x|_2^2=\langle\chi|\,(\,|x\rangle\langle x|\otimes\I_{\{0,1\}^n})\,|\chi\rangle$, which, if we measure $|\chi\rangle$ in the standard basis, is the probability of obtaining a snake whose initial vertex is in $Y_x$. Figure \ref{fig:difwaves} shows that the state $|\xi_{x_0,k_0,\sigma}\rangle$ evolves as a Gaussian-shaped wave packet moving with momentum $\tilde{\lambda}'(k_0)$ towards the glued part. The diffusion of the wave packet can be explained by the Heisenberg uncertainty principle: we cannot know both the position and the momentum
of the walk to an arbitrary precision.

\begin{figure}[htp]
\centering
\includegraphics[scale=0.9]{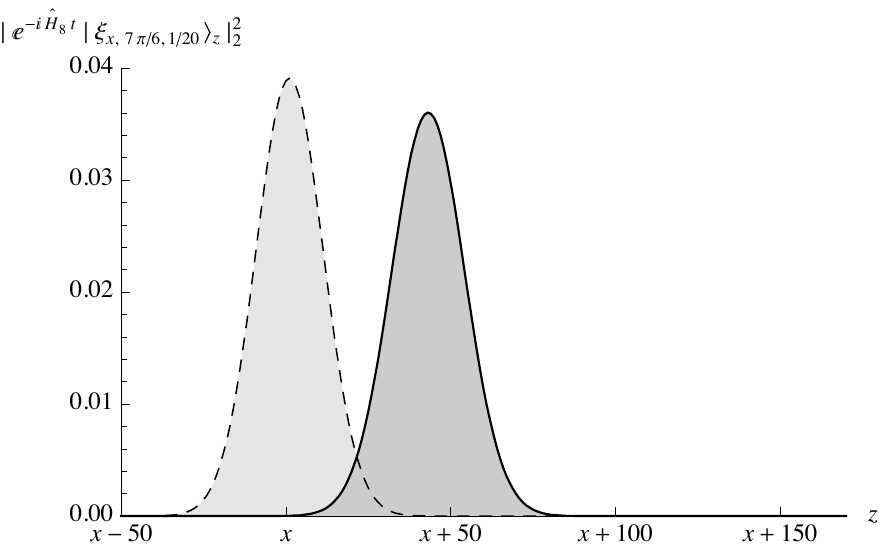}
\includegraphics[scale=0.9]{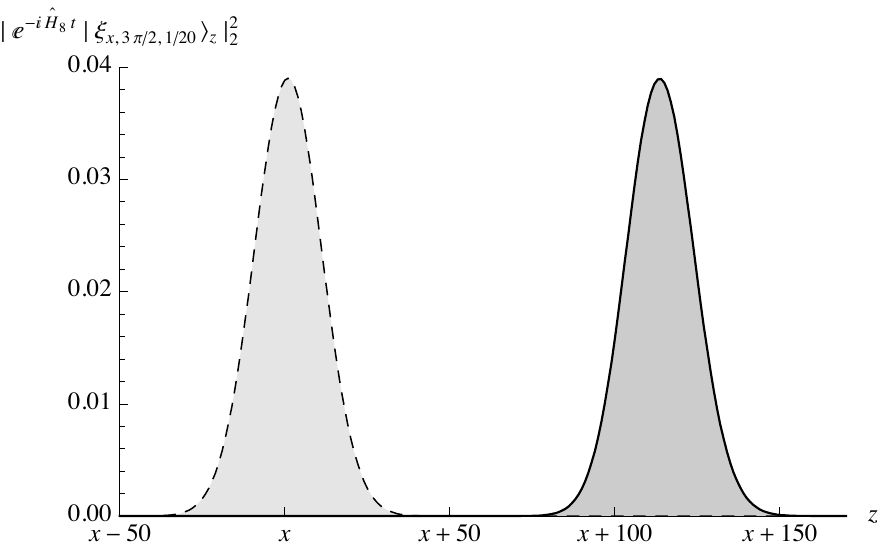}
\caption{The propagation of the wave packets $|\xi_{x_0,k_0,\sigma}\rangle$ for $k_0=\frac{7\pi}{6}$ (the left plot) and $k_0=\frac{3\pi}{2}$ (the right plot) in the case when $n=8$ and $\sigma=1/20$. The initial position of the wave packet is given with the dashed lines, and the position after time $t=100$ with the solid lines. Numerically we get $\hat{\lambda}'(\frac{7\pi}{6})\approx0.42$ and $\hat{\lambda}'(\frac{3\pi}{2})\approx1.13$.}
\label{fig:difwaves}
\end{figure}

For every $x\in\Z$ and $j\in\{0,1\}^n$ the span length of all snakes in $S(x,j)$ is the same, and it depends only on $j$. Therefore we can easily construct the unique observable $Q_n$ acting on $\C^{\{0,1\}^n}$ such that for every $|\chi\rangle\in\C^{\Z\times\{0,1\}^n}$, if we measure $|\chi\rangle$ in the standard basis and return the span length of the resulting snake, the expected value of this random variable is $\langle\chi|(\I_\Z\otimes Q_n)|\chi\rangle$. For $e^{-\ii \hat{H}_n t}|\xi_{x_0,k_0,\sigma}\rangle$ we have
\begin{equation*}
\begin{split}
\langle\xi_{x_0,k_0,\sigma}|e^{\ii \hat{H}_n t} (\I_\Z\otimes Q_n) e^{-\ii \hat{H}_n t}|\xi_{x_0,k_0,\sigma}\rangle
  & = \frac{1}{\erf(\pi/\sqrt{2}\sigma)}\int_{k_0-\pi}^{k_0+\pi} \frac{1}{\sigma\sqrt{2\pi}}e^{-\frac{(k-k_0)^2}{2\sigma^2}}\langle\hat{\psi}(k)| Q_n |\hat{\psi}(k)\rangle\,\dd k \\
  & \geq \min_{k} \langle\hat{\psi}(k)| Q_n |\hat{\psi}(k)\rangle.
\end{split}
\end{equation*}
Figure \ref{fig:span} shows a plot of numerically obtained values of $\langle\hat{\psi}(k)| Q_n |\hat{\psi}(k)\rangle$ for all values of even $n$ up to $20$. From this plot alone it is not completely clear how fast $\min_{k} \langle\hat{\psi}(k)| Q_n |\hat{\psi}(k)\rangle$ grows with $n$, but it seems to grow with the same rate as $\max_{k} \langle\hat{\psi}(k)| Q_n |\hat{\psi}(k)\rangle$.  For any $j\in\{0,1\}^n$ we have $\langle j| Q_n |j \rangle\geq|2|j|_H-n|$, where $|\cdot|_H$ is the Hamming norm, therefore for $j$ uniformly chosen at random we have $\langle j| Q_n |j \rangle\in\Omega(\sqrt{n})$. Because
\begin{equation*}
 |\hat{\psi}(\frac{\pi}{2})\rangle=-\ii\sqrt{\frac{2}{n+2}}\sum_{l=0}^{n/2}(|v_1\rangle\otimes|v_1\rangle)^{\otimes l}\otimes(|v_0\rangle\otimes|v_0\rangle)^{\otimes n/2-l},
\end{equation*}
 which is a quite balanced superposition over the vectors of $\{|j\rangle\,:\,j\in\{0,1\}^n\}$, we conjecture that $\min_{k} \langle\hat{\psi}(k)| Q_n |\hat{\psi}(k)\rangle\in\Omega(\sqrt{n})$.
\begin{figure}[htp]
\centering
\includegraphics[scale=1]{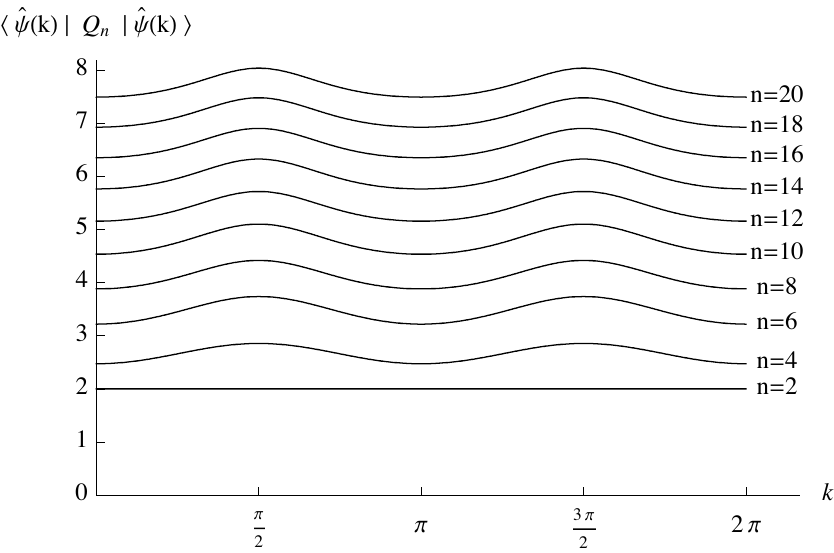}
\caption{The span length of the snake $|x\rangle\otimes|\hat{\psi}(k)\rangle$ for an arbitrary $x\in\Z$.}
\label{fig:span}
\end{figure}

\section{Evidence against efficiency of the algorithm} \label{app:unfort}

Consider the walk on the expanded glued trees graph started as the wave packet $|\xi_{x_0,k_0,\sigma}\rangle$, which moves towards the glued trees, and consider the maximal probability that, if at any point during the walk we make a measurement, we obtain a snake which contains vertices from both $Y_{-m}$ and $Y_{m+1}$. We are interested to choose $n\in O(\poly N)$ such that this probability is at least inverse-polynomially large in $N$ for $m=N$. It is not clear whether such choice is possible, and here we present numerical data suggesting that $n$ we need to choose may me very large, if not even exponentially large in $N$.

We introduced the vector $|\mu(k)\rangle$ in Hypothesis \ref{hyp:evec} as an eigenvector of $\tilde{H}_n$. Unless its coefficients $\alpha_{x,j}(k)$ behave very unexpectedly and $\arg'(\alpha_{x,j}(k))\ll -n$, for $x_0\ll -n$ and $\sigma\in\Omega\left(\frac{-1}{n+x_0}\right)$ the state $|\xi_{x_0,k_0,\sigma}\rangle$  can be well approximated by a superposition of states $|\mu(k)\rangle$, specifically, 
\begin{equation*}
|\xi_{x_0,k_0,\sigma}\rangle
\propto
\int_{k_0-\pi}^{k_0+\pi}
e^{-\frac{(k-k_0)^2}{4\sigma^2}}
e^{-\ii kx_0}|\tilde{k}\rangle\otimes|\hat{\psi}(k)\rangle\,\dd k
 \approx \frac{1}{\sqrt{2\pi}}
\int_{k_0-\pi}^{k_0+\pi}
e^{-\frac{(k-k_0)^2}{4\sigma^2}}
e^{-\ii kx_0}|\mu(k)\rangle\,\dd k.
\end{equation*}
Because $|\mu(k)\rangle$ is an eigenvector of $\tilde{H}_n$ for every $k\in\R$, the state $e^{-\ii \tilde{H}_n t}|\xi_{x_0,k_0,\sigma}\rangle$ can be well approximated by a superposition of states $|\mu(k)\rangle$ too. Therefore, if for all $k\in\R$ the state $|\mu(k)\rangle$ has a small overlap with the space of snakes connecting vertices from both $Y_{-N}$ and $Y_{N+1}$, so does $e^{-\ii \tilde{H}_n t}|\xi_{x_0,k_0,\sigma}\rangle$. Now, let us numerically calculate the coefficients $\alpha_{x,j}(k)$ and see what we obtain if we ``measure'' $|\mu(k)\rangle$ in the standard basis (note that $|\mu(k)\rangle$ has an infinite Euclidean norm). To be precise, we will be interested in the ``probabilities'' $p_{x,a}(k)$ defined as follows. For even $a\leq n$ and $x\in\Z$ and for odd $a\leq n$ and $x\in\Z+\frac{1}{2}$ let $P_{x,a}(k)$ be the projector to the space of all snakes having span $(x-\frac{a}{2},x+\frac{a}{2})\in\Z^2$. Let $p_{x,a}(k)=\langle\mu(k)|P_{x,a}(k)|\mu(k)\rangle$. We numerically calculated these values for $n=10$ and Figure \ref{fig:four} shows for every $a\in[2\,..\,n]$ how $p_{x,a}(\frac{3\pi}{2})$ depends on the value of $x$ (in this particular case $p_{x,1}(\frac{3\pi}{2})=0$ for all $x$).

\begin{figure}[htp]
\centering
\includegraphics[scale=0.93]{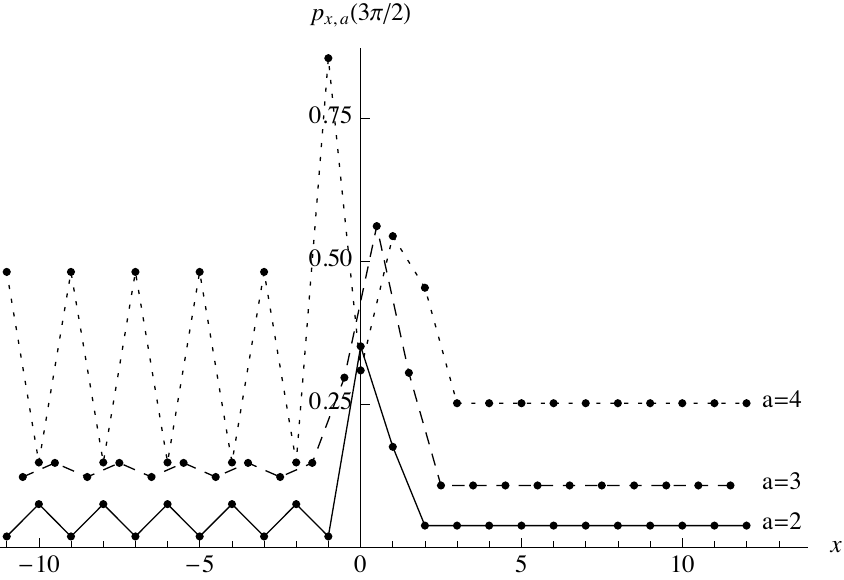}
\quad
\includegraphics[scale=0.93]{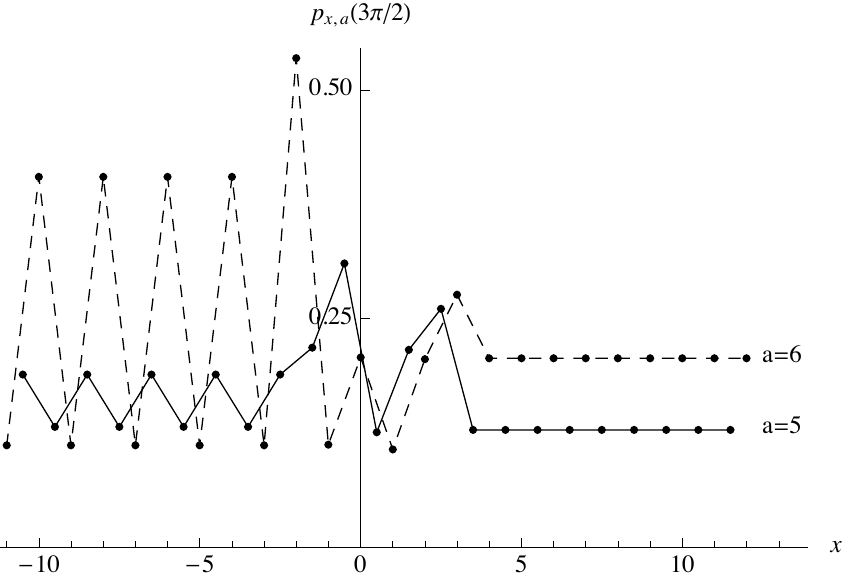}
\\
\medskip
\includegraphics[scale=0.93]{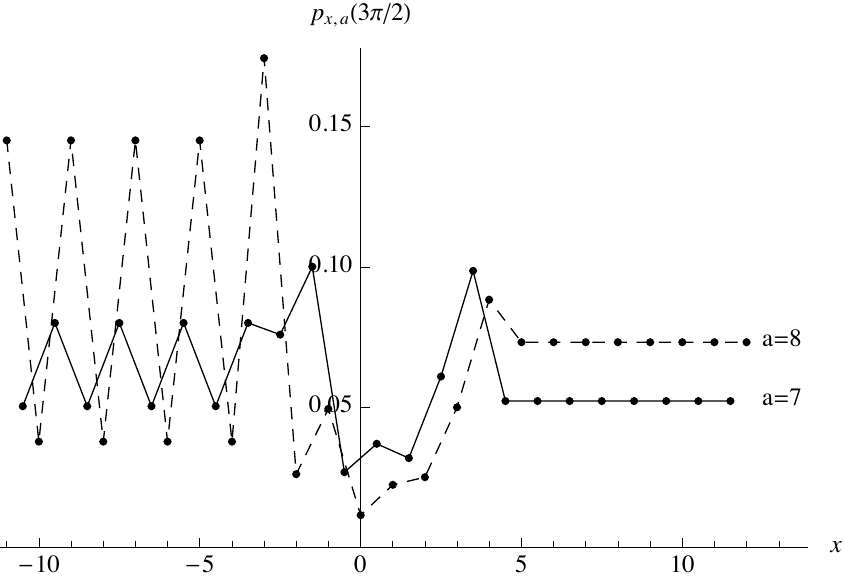}
\quad
\includegraphics[scale=0.93]{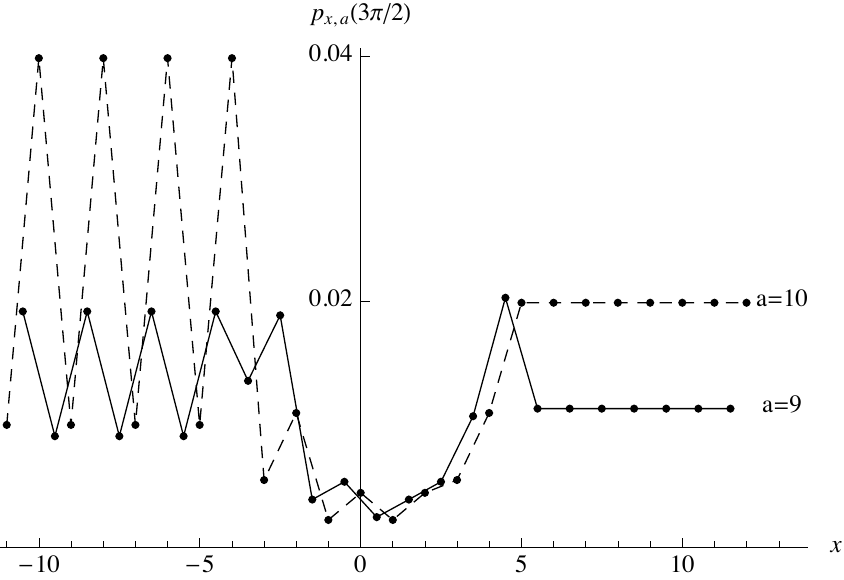}
\caption{The ``probability" of obtaining a snake having span $(x-\frac{a}{2},x+\frac{a}{2})$ if we ``measure" $|\mu(\frac{3\pi}{2})\rangle$. Note: the scaling of the plots differs.}
\label{fig:four}
\end{figure}

A similar pattern as in Figure \ref{fig:four} is observed if we consider values of $p_{x,a}(k)$ for $k$ other than $\frac{3\pi}{2}$ or ``measure'' eigenvectors of $\tilde{H}_n$ corresponding to other $k$-dependent eigenvalues of $\hat{H}_{n,k}$, not only the median eigenvalue. This means that the overlap between these eigenvectors and the space of snakes having a small span length drastically increases close to the glued part, and decreases for snakes having a large span length. Therefore, the expected span length of wave packets will get smaller when they come close to the glued part, and will increase again once they have been reflected from or gone through it. However, it is not clear what pattern this relative reduction of the span length follows when $n$ increases.

\end{document}